\documentclass[11pt, letterpaper]{amsart}
\usepackage{latexsym,amsfonts,amsmath, amsthm, amssymb}
\usepackage[mathscr]{eucal}
\usepackage{bm}

\newcommand\blfootnote[1]{
\begingroup
\renewcommand\thefootnote{}\footnote{#1}
\addtocounter{footnote}{-1}
\endgroup
}

\newcommand{\BbbR}{{\mathbb{R}}}

\newcommand{\BbbC}{{\mathbb{C}}}
\newcommand{\BbbP}{{\mathbb P}}

\newcommand{\fk}{{\mathfrak{k}}}

\newcommand{\Kr}{{K_{\text{rat}}}}

\newcommand{\CR}{{\mathcal{R}}}
\newcommand{\CB}{{\mathcal{B}}}
\allowdisplaybreaks

\newcounter{smalllist}

\newtheorem{lemma}{Lemma}[section]
\newtheorem{prop}[lemma]{Proposition}
\newtheorem{coro}[lemma]{Corollary}
\newtheorem{theorem}[lemma]{Theorem}
\theoremstyle{definition}
\newtheorem{exmp}[lemma]{Example}
\newtheorem{definition}[lemma]{Definition}
\theoremstyle{remark}
\newtheorem{remark}[lemma]{Remark}

\let\Im=\undefined\DeclareMathOperator{\Im}{Im}

\let\llldots=\ldots
\def\ldots{\llldots{}}

\numberwithin{equation}{section}

\makeatletter
\@namedef{subjclassname@2020}{\textup{2020} Mathematics Subject Classification}
\makeatother

\begin{document}

\title[reflection maps]
{Reflection maps associated with involutions and factorization problems, and their Poisson geometry}

\author[L.-C.~Li, C.~Caudrelier]{Luen-Chau Li $^{\ast}$ and Vincent Caudrelier}
\address{Luen-Chau Li\\
        Department of Mathematics\\
        Pennsylvania State University\\
        University Park, PA 16802, USA}
\email{lxl22@psu.edu}
\address{Vincent Caudrelier\\
          School of Mathematics\\
          University of Leeds\\
          Leeds, LS2 9JT UK}
\dedicatory{Dedicated to the memory of Henry McKean}
\email{V.Caudrelier@leeds.ac.uk}
\date{\today}

\blfootnote{$^{\ast}$  Corresponding author}
 
\begin{abstract}   The study of the set-theoretic solutions of the reflection equation, also known as reflection maps, is closely related to that of the Yang-Baxter maps.
In this work, we construct reflection maps on various geometrical objects, associated with factorization problems on rational loop groups and involutions.  We
show that such reflection maps are smoothly conjugate to the composite of permutation maps, with corresponding reduced Yang-Baxter maps.  In the case
when the reduced Yang-Baxter maps are independent of parameters, the latter are just braiding operators.  We also study the symplectic and
Poisson geometry of such reflection maps.   In a special case, the factorization problems are associated with the collision of $N$-solitons of the $n$-Manakov
system with a boundary, and in this context the $N$-body polarization reflection map is a symplectomorphism.
\end{abstract}

\thanks{}
\maketitle

%%%%%%%%%%%%%%%%%%%%%%%%%%%%%%%%%%%%%%%%%%%%%%%%%%%%%%%%%%%%%%%%%%
\section{Introduction}
%%%%%%%%%%%%%%%%%%%%%%%%%%%%%%%%%%%%%%%%%%%%%%%%%%%%%%%%%%%%%%%%%%%%%

The reflection equation, which is a close companion of the Yang-Baxter equation (YBE) \cite{Y, Ba}, first arose in the context of factorized scattering on a half-line \cite{C}.  It is
an important equation in the study of quantum integrable systems with non-periodic boundary conditions \cite{Sk}.   The YBE, as is well-known, is related to a multitude
of topics \cite{Y, B, J, T, JS, Tur, S, KL, Z}.    In particular,  any solution of YBE gives rise to  linear 
representation of the braid group $B_n$ \cite{K}.     Likewise, the reflection equation is connected with various branches of mathematics and physics (see, for example, \cite{C, Sk, KS, MR, DM, G, RV, AV, BS, BK, W, Xu}).   And in the same vein as YBE, the reflection equation deals with representation of the generalized
braid group $B_{1,n},$ which can be regarded as a subgroup of $B_{n+1}$ consisting of braids with one frozen strand (see, for example,
\cite{Ch, Sch, Lam}).

In the early 1990s, Drinfeld posed the problem of finding set-theoretic solutions of the Yang-Baxter
equation \cite{Dr2}.   More precisely, given a set $X,$ the problem is to find invertible maps $R:X\times X\longrightarrow X\times X$
satisfying
\begin{equation}\label{1.1}
R_{12}R_{13} R_{23} = R_{23}R_{13}R_{12}
\end{equation}
where $R_{ij}$ denotes the map from $X\times X\times X$ to itself which acts as $R$ on the $i$-th and $j$-th component
and as the identity on the other component.   This problem has led to numerous works (see, for example, \cite{WX, ESS, LYZ, GV, V, R, APT1, T, L1, CGP, ABS, HJN}).
Of these, the papers \cite{GV, V, APT1, T} are connected with soliton collisions in multi-component integrable PDEs, and such maps are referred to
as Yang-Baxter maps in \cite{V}.    It should be pointed out that in some cases in \cite{L1}, solutions of the YBE are not necessarily defined on a product
space, but we continue to call such maps Yang-Baxter maps, and this is the usage which we are following here.   In a nutshell, the YBE or its set-theoretic version is a condition which ensures the
factorization property.    In the context of the $n$-Manakov system (a.k.a.~vector NLS) on the line \cite{Man, APT2} (the $n=2$ case is due to Manakov in \cite{Man}), the vector solitons have internal degrees of
freedom, called polarizations.   Colliding solitons alter each other's polarization states, which is what makes colliding solitons of interest in optical computing
\cite{JSS}.   Here the factorization property means that an $N$-soliton collision process can be factorized into a nonlinear superposition of $N(N-1)/2$ pairwise collisions
in an arbitrary order, and the YBE ensures that all these possibilities gives the same result \cite{APT1, T}.  By studying the $n$-Manakov system on a half-line, with Robin boundary condition or 
mixed Dirichlet/Neumann condition at $x=0,$  the authors in \cite{CZ2} showed that there is also factorization in 
the presence of a boundary, and were led to introduce a set-theoretic version of the (parametric) reflection equation.   In recent years, the study of
set-theoretic solutions of the reflection equation, which are called reflection maps, and their relations to Yang-Baxter maps,
have been the subject of several studies \cite{CCZ, KO, dC, DS, SVW, LV, D}.

Our initial motivation for this work is to study the Poisson properties, if any, of the parametric reflection map which arises in \cite{CZ2},
where the authors are studying the interaction of $N$-solitons of the $n$-Manakov system on the half-line $x\geq 0$ with the boundary  at $x=0.$
Following \cite{APT2}, recall that the $n$-component Manakov system is given by the equation
\begin{equation}\label{1.2} 
i \bm{q}_t = \bm{q}_{xx} + 2\|\bm{q}\|^2 \bm{q},
\end{equation}
where $\bm{q}$ is a $\mathbb{C}^n$-valued function and $\|\bm{q}\|=(\bm{q}^{*} \bm{q})^{1/2}$ is the Euclidean norm of $\bm{q}.$ 
In \cite{CZ2}, the authors consider (1.2) on the half-line $x\geq 0$ and impose the following boundary conditions at $x=0:$
\newline
(a)  Robin boundary conditions, of the form
\begin{equation}\label{1.3}
\bm{q}_x(0, t) -2\alpha \bm{q}(0,t) =0,\,\,\, \alpha\in \mathbb{R},
\end{equation}
or 
\newline
(b) mixed Dirichlet/Neumann boundary conditions, given by
\begin{equation}\label{1.4}
\begin{aligned}
& q_j(0, t) =0, \,\,\, j\in S \subset \{1,\cdots, n\},\\
& q_{jx} (0, t) =0, \,\,\, j\in \{1,\cdots, n\}\setminus S.\\
\end{aligned}
\end{equation}
These boundary conditions are not arbitrary, but were derived earlier in \cite{CZ1},  in which the authors showed that a nonlinear mirror image method 
\cite{Kha, BT, BH} can be used to construct an inverse scattering transform for the half-line problem with that of a full-line problem.   As a consequence, 
one can obtain the $N$-soliton solution of the half-line problem with the above boundary conditions 
as the restriction to $x>0$ of a $2N$-soliton solution of a full-line problem provided that the norming constants and the
poles $\alpha_j$  satisfy appropriate mirror symmetry conditions which are dependent on the boundary conditions.  This is worked out in \cite{CZ2}
and the reflection map is precisely the map which describes the change in the polarization vector of a $1$-soliton when it
interacts with the boundary.   Note that in using the nonlinear mirror image method mentioned above, the collision of a $1$-soliton
with the boundary at $x=0$ becomes identified with the collision of the $1$-soliton with its `mirror' soliton.   Since soliton collision problems
correspond to refactorization problems associated with simple elements in rational loop groups, the kind of refactorization problems we consider in this connection
will have some special structure, or symmetry.   

Motivated by what we described in the previous paragraph, our goal in this work is to construct set-theoretic solutions of the reflection equation, or reflection maps, for a variety of geometric objects, and to study their symplectic and Poisson geometry.  The heuristic reason why the Yang-Baxter maps in \cite{L1}, and the reflection maps we consider here should have some interesting symplectic/Poisson
geometry is the following.   The $n$-Manakov system, and more generally multi-component integrable PDEs, are infinite dimensional Hamiltonian systems.  By general arguments, the dynamics of the multi-soliton solutions of such equations is expected to give rise to canonical maps on their respective phase spaces.   In the case of many scalar integrable PDEs, this is well-known (see, for example, \cite{FT}).   For a recent nontrivial example connected with the Benjamin-Ono equation, we refer the reader to \cite{Sun}.   As explained above, the dynamics of the $n$-Manakov system on the line can be described by  Yang-Baxter maps, and on the half-line with integrable boundary conditions by  reflection maps (up to asymptotic velocities and phase shifts).   It is therefore not only natural, but also a fundamental question to investigate if such maps have symplectic/Poisson properties with respect to
some symplectic/Poisson structures.   This is the overarching principle in \cite{L1} and in the present work.
Thus our main result here is that we establish, for the first time, the symplectic/Poisson nature
of our reflection maps, at the level of projectors, at the level of complex projective spaces, and at the level of Poisson Lie groups.  Since our reflection maps are associated with refactorization problems with some symmetry, associated with involutions,
it is natural to consider the reduction of the Yang-Baxter maps in \cite{L1}, which are diffeomorphisms on the graphs of
the involutions.  As it turns out, our reflection maps are smoothly conjugate to the composite of permutation maps, with
corresponding reduced Yang-Baxter maps.   In the case when the reduced Yang-Baxter maps are independent of parameters,
the latter are just braiding operators.
 This relationship opens up an approach to investigate the symplectic/Poisson geometry of the reflection maps, by means
 of reduction to symplectic submanifolds or reduction to Dirac submanifolds \cite{L2}, starting with the results in \cite{L1}.
\vskip .1in
The paper is organized as follows.   In section 2, we assemble some of the basic facts which we will use in
this work from several domains.   First of all, we will summarize some of the results in \cite{L1} on refactorization
problems in the rational loop group $K_{\text{rat}},$ and the symplectic and Poisson geometry of Yang-Baxter
maps,  as they provide the starting point of this investigation.   Then we will give the basic facts 
on the notion of Dirac submanifolds \cite{X},  and the method of Dirac reduction \cite{L2}.  In the context of
our present work, we are mainly dealing with the case in which the symplectic submanifolds and Dirac
submanifolds are given by the stable loci of symplectic/Poisson involutions.   In section 3, we begin
by introducing the class of involutions on $K_{\text{rat}}$ which we consider in this work.   Since what
we are doing here is motivated by the study of the $n$-Manakov system on the half-line with
$\mathbb{C}^n$-valued solutions $\mathbf{q},$   we must
include at least the two kinds of involutions which are implicit in \cite{CZ2}.    
To cut the story short, the two kinds of involutions corresponds to the two distinct kinds of boundary
conditions (a) and (b) described above.
As the reader will see, case (a) is not really of interest, and the class of involutions
which we will consider in this work, at the level of loops in the rational loop group $K_{\text{rat}}$ 
(we will also consider involutions at the level of projectors, at the level of $\mathbb{C}\mathbb{P}^{n-1}$ or abstractly on a Lie group), is given by
\begin{equation}\label{1.5}
\sigma: K_{\text{rat}}\longrightarrow K_{\text{rat}},\,\, \sigma(g)(z) = Ug^{*}(-\overline{z}) U^{*},
\end{equation}
where $U$ is a Hermitian unitary matrix, and the special case with
\begin{equation}\label{1.6}
U= I_{S} =  \text{diag}(d_1,\cdots, d_n), \,\,\text{where}\,\, d_i = \begin{cases}  1 & \,\text{if}\,\, i\in S\subset \{1,\cdots, n\},\\
-1 & \,\text{if}\,\, i\notin S,
\end{cases}
\end{equation}
is what corresponds to case (b) with mixed Dirichlet/Neumann boundary conditions.  Thus what 
we consider here is way beyond what we need to understand the case where $U= I_{S}.$   Note that if we restrict $\sigma$
to simple elements $g_{\alpha, P},$ of the form
\begin{equation}\label{1.7}
g_{\alpha, P}(z) = I + \frac{\overline{\alpha}-\alpha}{z-\overline{\alpha}}P,
\end{equation}
where $P$ is an $n\times n$ Hermitian projector of rank $k,$ $1\leq k\leq n-1,$  then we obtain an induced map
\begin{equation}\label{1.8}
(\alpha, P)\mapsto  (\tau(\alpha), c_{U,k}(P)),\,\,\, \tau(\alpha) = -\overline{\alpha}, c_{U,k}(P) = UPU^{*},
\end{equation}
which is also an involution and indeed,  this is what we will be working with in section 3.   In \cite{L1}, the author
showed that the parametric Yang-Baxter map $R^{k,k}(\alpha_1, \alpha_2)$ (to be defined below) is a symplectomorphism
on $P(n)_k\times P(n)_k,$  where $P(n)_k$ is the set of $n\times n$ Hermitian projectors of rank $k.$
The starting point of our investigation in that section consists of studying the restriction of 
$R^{k, k}(\tau(\alpha), \alpha)$ to the graph of $c_{U,k},$ which we denote by $\mathcal{G}_{U,k}.$   Indeed, 
$R^{k,k}(\tau(\alpha), \alpha)\mid_{\mathcal{G}_{U,k}}$ maps $\mathcal{G}_{U,k}$ to itself.
Thus we have the induced diffeomorphism $R^{k,k}_{\text{red}}(\tau(\alpha), \alpha)$ on 
$\mathcal{G}_{U,k}.$   As it turns out, $\mathcal{G}_{U,k}$ is a symplectic submanifold of
$P(n)_k\times P(n)_k,$ and the braiding operator corresponding to $R^{k,k}_{\text{red}}(\tau(\alpha), \alpha)$
is smoothly conjugate to the parametric reflection map $B(\alpha).$  This is the path we take in 
showing that $B(\alpha)$ is a symplectomorphism.    By using the partial action $\xi$ associated
with the refactorization problem in Theorem 5.13 (a) of \cite{L1}, and consideration related to 
the method of nonlinear mirror images (see the proof in Theorem 3.3 of \cite{CZ2}), 
we also show that $B(\alpha)$ satisfies the parametric reflection equation.

In section 4, we specialize to the rank $1$ case, in which we describe our results at the level
of complex projective space $\mathbb{C}\mathbb{P}^{n-1}.$   Here the involution is given
by the map
\begin{equation}\label{1.9}
\widetilde{c}_{U}: \mathbb{C}\mathbb{P}^{n-1}\longrightarrow \mathbb{C}\mathbb{P}^{n-1}, \,\, [p]\mapsto [Up].
\end{equation}
 In the particular case where
$U = I_{S},$  the corresponding parametric reflection map is what appeared in \cite{CZ2}
and describes the change in polarization when a $1$-soliton solution of the $n$-Manakov
system is reflected by the boundary at $x=0.$   Motivated by the interaction of $N$-soliton solutions
with the boundary at $x=0$ in this context,  we introduce the $N$-body {\sl polarization reflection
map} corresponding to a general Hermitian unitary matrix $U,$ and we show that this map is a symplectomorphism.   We wrap up that section with an explanation of the physical meaning of
the $N$-body polarization reflection map, in the context of the $n$-Manakov system on
the half-line with mixed Dirichlet/Neumann boundary conditions at $x=0.$   We also point out
the relationship between the full polarization scattering map, and that of the $N$-body
polarization reflection map.

In section 5, the last section of this work, we begin by proving two abstract results in the context of a Poisson Lie group $G.$  Here a number of assumptions must be
made.    As shown by the work in \cite{L1}, the Yang-Baxter maps associated with refactorization problems in $K_{\text{rat}}$ are not defined everywhere
on the product $K_{\text{rat}}\times K_{\text{rat}},$ and are associated with {\it partial} actions.  (This is in contrast to what is assumed in \cite{LYZ}.)
Thus we must postulate the existence of a left partial group action $\xi: G\ast G\longrightarrow G$ and a right partial group action  $\eta: G\ast G\longrightarrow G$ 
which are compatible in the sense that
\begin{equation}\label{1.10}
gh = \xi_{g}(h)\eta_{h}(g)\,\,\,\text{for all}\,\,\, (g,h) \in G\ast G,
\end{equation}
where $G\ast G$ is assumed to be an open submanifold of $G\times G.$   In addition, we have to make several assumptions on the domain of $\xi_{g}$
and $\eta_{g}$ for $g\in G.$   Of course, such assumptions are vacuous in case $\xi$ and $\eta$ are genuine actions.   On the other hand, we  have to postulate the existence of a Poisson involution $\sigma$ which is also a Lie group anti-morphism satisfying some additional assumptions.    These assumptions have to do with the intersection of the graph of $\sigma$ with $G\ast G,$ as well as the way $\sigma$ interacts with 
the two partial group actions.   As a matter of fact, if we let $G^{\prime}(\sigma) := G(\sigma) \cap (G\ast G),$ where $G(\sigma)$ denote the graph of 
$\sigma,$ we have to assume that
\begin{equation}\label{1.11}
G^{\prime} := (\sigma, id_{G})^{-1} (G^{\prime}(\sigma))\neq \emptyset .
\end{equation}
Under the first set of assumptions, we show that the map
\begin{equation}\label{1.12}
R: G\ast G\longrightarrow G\ast G, \,\, (g, h) \mapsto (\eta_{h}(g), \xi_{g}(h))
\end{equation}
is a Yang-Baxter map, and moreover, is a Poisson diffeomorphism, when the open Poisson submanifold $G\ast G$ is equipped with the structure induced from
$G\times G.$   With the additional postulates on $\sigma,$  we show that $G(\sigma),$ the graph of $\sigma,$ is a Dirac submanifold of $G\times G,$ equipped
with the product structure.   Thus $G(\sigma)$ carries an induced Poisson structure  $\pi_{G(\sigma)}.$  On the other hand, we can push the Poisson structure
on $G$ forward to $G(\sigma)$ using the map $(\sigma, id_{G})$ so that it becomes a Poisson diffeomorphism, when its codomain is equipped with the pushforward structure.  It is miraculous that this pushforward structure is precisely $2 \pi_{G(\sigma)}.$    Note that the reflection map in this setting, which we denote by $\mathbf{B},$
is only defined on the open submanifold $G^{\prime}$ of the Poisson Lie group $G,$  but we can show that
it is smoothly conjugate to the braiding operator, of the reduced Yang-Baxter map $R_{\text{red}}: G^{\prime}(\sigma)\longrightarrow G^{\prime}(\sigma).$
Combining this with the results on the various Poisson structures, we  conclude that $\bold{B}$ is a Poisson diffeomorphism.   With suitable restrictions, we
can also show that $\mathbf{B}$ satisfies the reflection equation.   Finally, we conclude the section by applying the general results to $K_{\text{rat}}.$  There
are various conditions to check, see, in particular, Proposition 5.9 and Proposition 5.10.

We end the paper with a conclusion section in which we discuss what has been achieved as well as giving some perspectives on future directions.

%%%%%%%%%%%%%%%%%%%%%%%%%%%%%%%%%%%%%%%%%%%%%%%%%%%%%%%%%%%%%%%%%%%%%%%%%%%%%%%%%%%%%%
\section{Preliminaries}
%%%%%%%%%%%%%%%%%%%%%%%%%%%%%%%%%%%%%%%%%%%%%%%%%%%%%%%%%%%%%%%%%%%%%%%%%%%%%%%%%%%%%%%

In this section, we will first of all collect a number of results from \cite{L1} which serves as the starting point of our analysis in this work.   We will also recall the notion of Dirac submanifolds \cite{X} and the method of Dirac reduction \cite{L2} which will be used in the section on reflection maps and Poisson Lie groups.  

\subsection{Parametric Yang-Baxter maps}  

We begin by introducing the loop groups which play an essential role in \cite{L1}.    To do so, let $\BbbC\BbbP^{1}= \BbbC\cup \{\infty\},$ $\Omega_{+}=\BbbC,$
$\Omega_{-} = \mathcal{O}_{\infty},$ a neighborhood of $\infty$ invariant
under complex conjugation.  Also, let $U(n)$ be the unitary group, and denote its Lie algebra by $u(n).$  Following Terng and Uhlenbeck \cite{TU},
we introduce the loop group
\begin{equation}\label{2.1}
\begin{aligned}
& K=\{g:\Omega_{+}\cap\Omega_{-}\longrightarrow GL(n,\BbbC)\mid g\,\,\text{is holomorphic, and}\,\, g(\bar{z})^{*}g(z)=I, \,\, \text{for}\\
& \quad \quad\,\,\, \text{all}\,\, z\in \Omega_{+}\cap \Omega_{-}\}.\\
 \end{aligned}
\end{equation}
For a matrix loop $g$ which is holomorphic in $\mathcal{U}\subset \BbbC\BbbP^{1}$  satisfying $g(\bar{z})^{*}g(z)=I$ for all $z\in \mathcal{U},$
we say that $g$ is $u(n)$-holomorphic on $\mathcal{U}.$   Then we have the following Lie subgroups of $K$:
\begin{equation}\label{2.2}
\begin{aligned}
& K_{+} =\{g\in K\mid g\,\,\text{extends}\,\,u(n)-\text{holomorphically to}\,\,  \Omega_{+}\},\\
& K_{-} =\{g\in K\mid g\,\,\text{extends}\,\,u(n)-\text{holomorphically to}\,\,  \Omega_{-}, g(\infty)=I\},\\
& \Kr =\,\,\text{subgroup of rational maps}\,\, g\in K_{-}.\\
\end{aligned}
\end{equation}
In $\Kr,$ there are some special elements that are of basic importance.  To introduce 
these elements, let $\mathcal{H}(n)$ denote the set of $n\times n$ Hermitian matrices, 
and let
\begin{equation}\label{2.3}
P(n) = \{ P\in \mathcal{H}(n)\mid P^2 =P \}.
\end{equation}
Then associated to each $\alpha\in \BbbC\setminus \BbbR,$ and $P\in P(n),$  is the simple element
\begin{equation}\label{2.4}
g_{\alpha, P}(z) = I + \frac{\overline{\alpha}-\alpha}{z-\overline{\alpha}}P.
\end{equation}
These simple elements are  known as dressing factors in the work of Zakharov and Shabat \cite{ZS}, they are
called Blaschke-Potapov factors in \cite{FT} (see \cite{P}).
\begin{theorem}[\cite{U}] \label{T:2.1}
(a)  $g_{\alpha, P}\in \Kr.$
\newline
(b)  $\Kr$ is generated by the simple elements, i.e., every $g\in \Kr$ can be factorized into a product of simple elements.
\end{theorem}
Note that if we drop the reality condition $g(\overline{z})^* g(z) =I$ in $K_{\text{rat}},$ the result is $G_{\text{rat}},$ which is called the {\sl full rational 
loop group} in \cite{Goe}.   Clearly, we have the (involutive) automorphism $G_{\text{rat}}\longrightarrow G_{\text{rat}}: g(z)\longrightarrow (g(\overline{z})^*)^{-1}$
and $K_{\text{rat}}$ is the fixed point subgroup of this automorphism.   The reason why we consider $K_{\text{rat}}$ is due to the fact that we
are geared towards the $n$-Manakov system here, whose Lax operator in the zero curvature representation has certain symmetries (see, e.g. \cite{APT2}), but this is by no means necessary in the study Yang-Baxter maps.   Indeed, in \cite{L1}, $G_{\text{rat}}$ is also considered.  In this connection, we refer the reader to \cite{Mi} for a systematic study of various reductions in  zero curvature representations.

We will equip $\mathcal{H}(n)\simeq u(n)^*$ with the Lie-Poisson structure, where the identification is through the map
\begin{equation}\label{2.5}
\kappa: \mathcal{H}(n)\longrightarrow u(n)^{*},\,\,\, \kappa(A)(B) = -\sqrt{-1} \,\text{tr}\,(AB).
\end{equation}
Then $P(n)$ is a Poisson submanifold of  $\mathcal{H}(n).$   If for $1\leq k\leq n-1,$ we let
\begin{equation}\label{2.6}
P(n)_k =\{ P\in P(n)\mid \text{tr}\, P =k\}
\end{equation}
then $P(n)_k$ is nothing but the coadjoint orbit of the unitary group $U(n)$ through the point $E_k,$ defined by
the following formula:
\begin{equation}\label{2.7}
E_k = \begin{pmatrix} I_k & 0\\
                    0  & 0_{n-k}\\
\end{pmatrix} .
\end{equation}
Hence there is a standard symplectic structure on $P(n)_k,$ given by
\begin{equation}\label{2.8}
\omega_{E_k}(P)([X,P],[Y,P]) = \sqrt{-1} \text{tr}\,P[X,Y], \,\,P\in P(n)_k,\, X,Y\in u(n).
\end{equation}

\begin{theorem}[\cite{L1}]\label{T:2.2}
(a) For given $g_{\alpha_i, P_i}$ with $\alpha_i\in \BbbC\setminus \BbbR,$ $i=1,2,$ the
refactorization problem
\begin{equation}\label{2.9}
g_{\alpha_1, P_1}\, g_{\alpha_2, P_2} = g_{\alpha_2, \widetilde{P}_2}\, g_{\alpha_1, 
\widetilde{P}_1}
\end{equation}
has a unique solution if 
\begin{equation}\label{2.10}
\{\alpha_1, \overline{\alpha}_1\}\cap \{\alpha_2, \overline{\alpha}_2\} =\emptyset.
\end{equation}
In this case, the projections $\widetilde{P}_1$ and $\widetilde{P}_2$ are given by
\begin{equation}\label{2.11}
\widetilde{P}_i = \phi P_i \phi^{-1},
\end{equation}
where
\begin{equation}\label{2.12}
\phi = (\alpha_2 -\overline{\alpha}_1)I + (\overline{\alpha}_2 -\alpha_2)P_2  + (\overline{\alpha}_1-\alpha_1) P_1,
\end{equation}
and we define 
\begin{equation}\label{2.13}
R(\alpha_1, \alpha_2)(P_1, P_2) = (\widetilde{P}_1, \widetilde{P}_2).
\end{equation}
(b)  Let $\alpha_i\in \mathbb{C}\setminus \mathbb{R},$ $i=1,2,3$ satisfy
\begin{equation}\label{2.14}
\{\alpha_i, \overline{\alpha}_i\}\cap \{\alpha_j, \overline{\alpha}_j\}= \emptyset, i\neq j,
\end{equation} 
then $R_{ij}(\alpha_i, \alpha_j)$ satisfy the parametric Yang-Baxter equation
\begin{equation}\label{2.15}
R_{12}(\alpha_1, \alpha_2)R_{13}(\alpha_1, \alpha_3) R_{23}(\alpha_2,\alpha_3) = R_{23}(\alpha_2, \alpha_3) R_{13}(\alpha_1,\alpha_3)R_{12}(\alpha_1,\alpha_2)
\end{equation}
on $P(n) \times P(n)\times P(n).$                                                                                                                                                                                                                                                                                                                                                                                                                                                                                                                                                                                                                                                                                                                                                                                                                                                                                                                                                    
(c)  Let $\alpha_1,\alpha_2$ be as in part (a), and denote by $\{\cdot, \cdot\}_{P(n)}$  the bracket on the Poisson submanifold $P(n)$ 
of $\mathcal{H}(n)$ equipped with the Lie-Poisson structure. Consider
$P(n) \times P(n)$ with the product Poisson structure, where the first copy of $P(n)$ is equipped
with $(-2\Im\alpha_1)\{\cdot, \cdot\}_{P(n)},$ and the second copy of $P(n)$
is equipped with $(-2\Im{\alpha_2})\{\cdot, \cdot\}_{P(n)}.$  If we denote
this bracket by $\{\cdot, \cdot\},$ then the map
\begin{equation}\label{2.16}
R(\alpha_1,\alpha_2): (P(n)\times P(n), \{\cdot, \cdot\})\longrightarrow (P(n)\times P(n), \{\cdot, \cdot \})
\end{equation}
is a Poisson diffeomorphism.
Indeed, for any $1\leq k, \ell\leq n-1,$ if we let 
$R^{k,\ell}(\alpha_1,\alpha_2)= R(\alpha_1,\alpha_2)\mid P(n)_k\times P(n)_{\ell},$
and equip $P(n)_k\times P(n)_{\ell}$ with the
symplectic structure given by the $2$-form
\begin{equation}\label{2.17}
\omega^{\prime}_{\alpha_1,\alpha_2} = (-2\Im \alpha_1)\omega_{E_k}
\oplus (-2\Im \alpha_2)\omega_{E_{\ell}},
\end{equation}
then 
\begin{equation}\label{2.18}
R^{k,\ell}(\alpha_1,\alpha_2)
: (P(n)_k\times P(n)_{\ell}, \omega^{\prime}_{\alpha_1,\alpha_2})\longrightarrow
(P(n)_k\times P(n)_{\ell}, \omega^{\prime}_{\alpha_1,\alpha_2})
\end{equation}
is a symplectomorphism.
\end{theorem}
While the map $R(\alpha_1,\alpha_2)$ is defined by the refactorization problem in \eqref{2.9}, we define the map
$R_{21}(\alpha_2, \alpha_1)$ on $P(n)\times P(n)$ by 
\begin{equation}\label{2.19}
R_{21}(\alpha_2, \alpha_1)(P_1, P_2) = (\widetilde{P}_1, \widetilde{P}_2), \,\,\,\text{where}\,\,\, g_{\alpha_2, P_2} g_{\alpha_1, P_1} = g_{\alpha_1, \widetilde{P}_1} g_{\alpha_2, \widetilde{P}_2},
\end{equation}
and of course, this gives rise to maps $R^{k,\ell}_{21}(\alpha_2, \alpha_1)$ from $P(n)_k \times P(n)_{\ell}$ to itself.   Clearly, we have
\begin{equation}\label{2.20}
R_{21}(\alpha_2,\alpha_1)R(\alpha_1,\alpha_2) = id_{P(n)\times P(n)}
\end{equation}
and moreover,
\begin{equation}\label{2.21}
R^{k,\ell}_{21} (\alpha_2, \alpha_1) = S_{k,\ell}^{-1}\circ R^{\ell, k}(\alpha_2, \alpha_1)\circ S_{k,\ell},
\end{equation}
where $S_{k,\ell}: P(n)_k\times P(n)_{\ell}\longrightarrow P(n)_{\ell}\times P(n)_k$ is the permutation map that sends $(P_1, P_2)$ to $(P_2, P_1).$

\subsection{The rank 1 case.}
\vskip .1in

The rank $1$ case, which corresponds to $k=1$ in \eqref{2.6}, is related to soliton collisions in the $n$-Manakov system. In that system, people
usually deal with the change of {\sl unit} polarization vectors \cite{APT1, APT2} and projectors of rank $1$ are constructed from
such vectors.    Since the {\sl change of unit polarization map} is a map from $S^{2n-1}\times S^{2n-1}$ into itself, such a map
cannot be symplectic as $S^{2n-1}$ is odd-dimensional.   This is the reason why the author is working with $\mathbb{C}\mathbb{P}^{n-1}$
instead of $S^{2n-1}$ in \cite{L1}.   Here is the result we will use in this work.

\begin{theorem}[\cite{L1}]\label{T: 2.3}
(a) Let $j_{\delta}$ be the map given by
\begin{equation}\label{2.22}
j_{\delta} :\mathbb{C}\mathbb{P}^{n-1}\longrightarrow P(n)_1, [p]\mapsto \frac{pp^*}{p^*p} = \pi_{[p]},
\end{equation}
then the pullback of $\omega_{E_1}$ under $j_{\delta}$ is the Fubini-Study $2$-form
\begin{equation}\label{2.23}
j_{\delta}^* \omega_{E_1} = \omega_{FS} = \frac{p^* dp\wedge dp^* d + (p^*p)dp^*\wedge dp}{(p^*p)^2} .
\end{equation}
(b) Equip $\mathbb{C}\mathbb{P}^{n-1}\times \mathbb{C}\mathbb{P}^{n-1}$ with the symplectic $2$-from
\begin{equation}\label{2.24}
\begin{aligned}
\Omega_{\alpha_1, \alpha_2} = & (j_{\delta}\times j_{\delta})^{*}( (\overline{\alpha}_1-\alpha_1)\omega_{E_1}\oplus
 (\overline{\alpha}_2-\alpha_2) \omega_{E_1})\\
= & (\overline{\alpha}_1-\alpha_1) \omega_{\text{FS}} \oplus (\overline{\alpha}_2-\alpha_2)
\omega_{\text{FS}},\\
\end{aligned}
\end{equation}
then the map 
\begin{equation}\label{2.25}
\begin{aligned}
& \widetilde{R} (\alpha_1,\alpha_2) : \mathbb{C}\mathbb{P}^{n-1} \times \mathbb{C}\mathbb{P}^{n-1}
\longrightarrow \mathbb{C}\mathbb{P}^{n-1}\times \mathbb{C}\mathbb{P}^{n-1},\\
& ([p_1],[p_2])\mapsto ([\phi_{\alpha}([p_1], [p_2]) p_1], [\phi_{\alpha}([p_1],[p_2]) p_2])
\end{aligned}
\end{equation}
is a symplectomorphism, where
\begin{equation}\label{2.26}
\phi_{\alpha}([p_1],[p_2]) = (\alpha_2-\overline{\alpha}_1) I + (\overline{\alpha}_2-\alpha_2) \pi_{[p_2]} + (\overline{\alpha}_1-\alpha_1)\pi_{[p_1]}.
\end{equation}
Moreover, we have the following formulas:
\begin{equation}\label{2.27}
\begin{aligned}
& \phi_{\alpha}([p_1],[p_2]) p_1 = (\alpha_2-\alpha_1) g_{\overline{\alpha}_2,\alpha_2, \pi_{[p_2]}}
(\alpha_1) p_1,\\
& \phi_{\alpha}([p_1],[p_2]) p_2 = (\overline{\alpha}_2-\overline{\alpha}_1) g_{\alpha_1,\overline{\alpha}_1, \pi_{[p_1]}}(\overline{\alpha}_2) p_2.\\
\end{aligned}
\end{equation}
\end{theorem}
In Section 4 of \cite{L1},  the author also has a result on the polarization scattering map.    For our purpose here, let us recall its definition since we will use it in
Section 4 of the present work.  We start with the fact that for the $n$-Manakov system, an $N$-soliton solution corresponding to distinct eigenvalues
\begin{equation}\label{2.28}
\alpha_j = \frac{1}{2}(u_j + iv_j), j=1,\cdots, N
\end{equation}
in the upper half-plane behaves asymptotically as a sum of $N$ one-soliton solutions as $t\to\pm \infty$:
\begin{equation}\label{2.29}
\bold{q}(x,t)\sim \sum_{j=1}^{N} q_{j}^{\pm}(x,t) p_j^{\pm},
\end{equation}
where $p_j^{\pm}$ are unit vectors in $\mathbb{C}^n,$ and $q_j^{\pm}(x,t)$  is parametrized by $u_j, v_j,$ $j=1,\cdots, N.$ Here $p_j^{-}$ (resp.~ $p_j^{+})$ is the unit asymptotic polarization vector of soliton $j$
before (resp.~after) all its collisions.   We call the map defined by
\begin{equation}\label{2.30}
\begin{aligned}
\mathcal{S}(\alpha_1,\cdots, \alpha_N) &: (\mathbb{C}\mathbb{P}^{n-1})^N\longrightarrow (\mathbb{C}\mathbb{P}^{n-1})^N \\
& ([p_1^{-}],\cdots, [p_N^{-})])\mapsto   ([p_1^{+}],\cdots, [p_N^{+})])\\
\end{aligned}
\end{equation}
 the {\sl polarization scattering map}.   In Section 4 below,  we will make use of
 it to explain the meaning of the $N$-{\sl body polarization reflection map} in the context of soliton-boundary interactions for the $n$-Manakov system.

\subsection{Yang-Baxter maps and $K_{\text{rat}}$}
\vskip .1in
Recall that a non-singular rational matrix function $A$ has as many poles and zeros in $\mathbb{C}\mathbb{P}^1 = \mathbb{C}\cup \{\infty\}$ \cite{Kai}.  If $A$
is such a rational matrix function, the divisor of $A$ is denoted by $(A) = (A)_0 - (A)_{\infty},$ where $(A)_0$ is the divisor of zeros and $(A)_{\infty}$
is the divisor of poles.    We are dealing with $g\in K_{\text{rat}},$ the rational loop group.   Since $g(\overline{z})^{*} g(z) = I,$ it follows that
the divisor of $g$ is of the form
\begin{equation}\label{2.31}
(g) = \sum_{i=1}^{\ell} n_i\cdot \alpha_i -\sum_{i=1}^{\ell} n_i \cdot \overline{\alpha}_i,
\end{equation}
where $\alpha_1,\cdots, \alpha_{\ell}$  (resp.~ $\overline{\alpha}_1,\cdots, \overline{\alpha}_{\ell}$)  are distinct
zeros (resp.~ poles) of $g$ with orders $n_1,\cdots, n_{\ell}.$   Conversely, given a divisor $D\in \text{Div}^0(\mathbb{C}\setminus\mathbb{R})$
satisfying $\overline{D} = -D,$ we let
\begin{equation}\label{2.32}
K_{\text{rat}} (D) = \{g\in K_{\text{rat}}\mid (g) = D\},
\end{equation}
and we denote by $\text{supp}\, D$ the support of $D.$
Note that in contrast to the $n=1$ case, a rational matrix function can have zero and pole at the same point (see Remark 2.3 in \cite{L1}).
In order to state the next result, we first recall 
the notion of partial group actions, which first appeared in the study of some $C^*$-algebras \cite{E} (see also
\cite{B, L1}).     
\begin{definition}\label{D:2.4}  Let $M$ be a smooth manifold, and let $G$ be a Lie group.  A left partial action of  $G$ on $M$ consists of a family $\{M_g\}_{g\in G}$
of subsets of $X$ and a family of bijections $\{\Phi_g: M_{g^{-1}}\longrightarrow M_g\}_{g\in G}$
satisfying the following conditions:
\newline
(a) $M_e= X,$ $\Phi_e = \text{id}_M,$
\newline
(b)  $\Phi_h^{-1}(M_{g^{-1}}\cap M_h)\subset M_{(gh)^{-1}},$  $g,h\in G,$
\newline
(c) $\Phi_g(\Phi_h(x)) = \Phi_{gh}(x)$\,\,  for each $x\in \Phi_h^{-1}(M_{g^{-1}}\cap M_h).$

We say that the left partial group action is smooth if 
\begin{equation}\label{2.33}
 G\ast M =\{(g,x)\in G\times M\mid x\in M_{g^{-1}}\}
 \end{equation}
 is a smooth submanifold of the product manifold $G\times M$ and the map
 \begin{equation}\label{2.34}
 \Phi: G\ast M\longrightarrow M, \,\,(g,x)\mapsto \Phi_{g}(x)
 \end{equation}
 is a smooth map.
\end{definition}
In a similar way, we can define right partial group action of $G$ on $M.$  We will make use of the following result established in \cite{L1}.
\begin{theorem}[\cite{L1}]\label{T: 2.5}  (a) Given $u\in \Kr(D),$ $v\in \Kr(D^{\prime}),$ where $D\neq D^{\prime}$ are divisors in $\text{Div}^0 (\BbbC\setminus\BbbR)$
satisfying $\overline{D} = -D,$ $\overline{D^{\prime}}= -D^{\prime},$ the refactorization problem 
\begin{equation}\label{2.35}
uv = \widetilde{v}\widetilde{u} 
\end{equation}
of finding $\widetilde{u}\in \Kr(D),$ $\widetilde{v}\in \Kr(D^{\prime})$ has a unique solution if  $\text{supp}\,D\cap \text{supp}\,D^{\prime}=\emptyset,$ in which case we write
\begin{equation}\label{2.36}
\widetilde{v} = \xi_{u}(v),\quad \widetilde{u} =\eta_{v}(u).
\end{equation}
Thus for each $u\in \Kr,$ $\xi_u$
is defined on 
\begin{equation}\label{2.37}
K^{u^{-1}}_{\text{rat}} = \{v\in \Kr\mid \text{supp}\, (u) \cap \text{supp}\, (v) =\emptyset \}= K^{u}_{\text{rat}}
\end{equation}
and it takes values in the same set.
Similarly, for given $v\in \Kr,$ $\eta_v$ is defined on $K^{v^{-1}}_{\text{rat}}= K^{v}_{\text{rat}}$  and takes values in the same set. 
\newline
(b)  Define
\begin{equation}\label{2.38}
\Kr \ast \Kr =\{ (u,v)\in \Kr \times \Kr \mid \text{supp} (u)\cap \text{supp} (v) =\emptyset \},
\end{equation}
then the map 
\begin{equation}\label{2.39}
\xi: \Kr \ast \Kr \longrightarrow \Kr, \,\,(u,v) \mapsto \xi_{u} (v)
\end{equation}
 is a left partial group action.   Similarly, the map
 \begin{equation}\label{2.40}
\eta: \Kr \ast \Kr\longrightarrow \Kr,\,\, (u,v) \mapsto \eta_{v} (u)
\end{equation}
 is a right partial group action.
\end{theorem}
Now introduce the map
\begin{equation}\label{2.41}
R: \Kr \ast \Kr\longrightarrow \Kr \ast \Kr, (u,v)\mapsto (\eta_{v}(u), \xi_{u}(v)).
\end{equation}
We also introduce
\begin{equation}\label{2.42}
K^{(3)}_{\text{rat}} =\{(u_1,u_2, u_3)\in \Kr\times\Kr\times \Kr\mid \text{supp} (u_i)\cap \text{supp} (u_j)=\emptyset, i\neq j\}.
\end{equation}
Then as a corollary of the theorem above, we have
\begin{coro}[\cite{L1}]\label{C: 2.6}  The map $R$ is a diffeomorphism satisfying the set-theoretical Yang-Baxter equation
\begin{equation}\label{2.43}
R_{12} R_{13} R_{23} = R_{23} R_{13} R_{12},
\end{equation}
where we interpret \eqref{2.43} as an equality of maps from $K^{(3)}_{\text{rat}}$ to $K^{(3)}_{\text{rat}}.$
\end{coro}
In order to describe the Poisson character of the map $R,$ let $\fk,$ $\fk_{\pm}$ be the Lie algebras of the loop groups $K,$ $K_{\pm}$ introduced in the first
subsection.  We equip
$\fk$ with the invariant pairing
\begin{equation}\label{2.44}
(X,Y)_{\fk} = \oint_{\gamma}  \text{tr} (X(z) Y(z))\, \frac{dz}{2\pi i}, \quad \gamma=\partial \mathcal{O}_{\infty}, \quad X,Y\in \fk.
\end{equation}
Since $\fk = \fk_{+} \oplus \fk_{-},$ we have the associated projection operators $\Pi_{\fk_{+}},$ $\Pi_{\fk_{-}}$
and 
\begin{equation}\label{2.45}
J= \Pi_{\fk_{+}} -\Pi_{\fk_{-}}
\end{equation}
is a skew-symmetric solution (w.r.t. $(\cdot, \cdot)_{\fk}$ ) of the modified Yang-Baxter equation (mYBE).   In order to introduce 
the Poisson structure on $K,$ it is necessary to restrict ourselves
to a subclass of functions in $C^{\infty}(K).$   Following the approach in
\cite{LN}, a function $\varphi\in C^{\infty}(K)$ is {\sl smooth at $g\in K$} iff 
there exists $D\varphi(g)\in \fk$ (called the right gradient
of $\varphi$ at $g$) such that 
\begin{equation}\label{2.46}
\frac{d}{dt}{\Big|_{t=0}} \varphi(e^{tX}g) = 
(D\varphi(g), X)_{\fk},\quad X\in \fk
\end{equation}
where $(\cdot, \cdot)_{\fk}$ is the pairing in \eqref{2.44}.
If $\varphi\in C^{\infty}(K)$ is smooth at $g$ for all 
$g\in K,$  then we say it is {\sl smooth on} $K.$ 
Note that the nondegeneracy of $(\cdot, \cdot)_{\fk}$ implies that 
the map
\begin{equation}\label{2.47}
i:\fk\longrightarrow \fk^{*},\quad X\mapsto (X,\cdot)_{\fk}
\end{equation} 
is an isomorphism 
onto a subspace of $\fk^{*}$ which we will call
the smooth part of $\fk^{*}.$
Thus $\varphi\in C^{\infty}(K)$ is smooth at $g$
iff $T_{e}^{*}r_{g} d\varphi(g)$ is in
the smooth part of $\fk^{*}$ and we can
define the left gradient of such a function at $g$ by 
\begin{equation}\label{2.48}
\frac{d}{dt}\Big|_{t=0} \varphi(ge^{tX}) = 
(D' \varphi(g), X)_{\fk},\quad X\in \fk.
\end{equation}
For each $g\in K,$ we will denote the collection of all 
smooth functions at $g$ by ${\mathcal F}_{g}(K)$ and we set
${\mathcal F}(K) =\cap_{g\in K} {\mathcal F}_{g}(K).$  With the above considerations, it is easy
to check that ${\mathcal F}(K)$ is non-empty and forms an algebra under ordinary
multiplication of functions.   Moreover, for $\varphi, \psi\in {\mathcal F}(K)$ and $g\in K,$ the expression
\begin{equation}\label{2.49}
\{\varphi, \psi\}_{J}(g) = \frac{1}{2}(J(D\varphi(g)), D\psi(g))_{\fk}-\frac{1}{2} (J(D'\varphi(g)), 
D'\psi(g))_{\fk} .
\end{equation}
defines $\{\varphi, \psi\}_{J}\in {\mathcal F}(K)$ (Proposition 5.3 of \cite{L1}, the proof is identical to that in \cite{LN}, Proposition 3.1)
and hence is a Poisson bracket on ${\mathcal F}(K).$
Hence $(K, \{\cdot, \cdot\}_{J})$ is a  coboundary Poisson Lie group, in the sense of Drinfeld \cite{Dr1}.
In \cite{L1}, we showed that $K_{\text{rat}}$ is a Poisson Lie subgroup of $(K, \{\cdot, \cdot\}_{J}).$  

 In order to state the next result, we will have to use the notion of Poisson group partial actions introduced in \cite{L1}.   We begin by recalling the notion of a coisotropic submanifold of a Poisson manifold \cite{CdSW}.
\begin{definition}\label{D:2.7}  A submanifold $C$ of a Poisson manifold $(M, \pi)$ is coisotropic if for each $m\in C,$ the annihilator
$T_mC^{\perp}\subset T_m^*M$  is isotropic, i.e., 
\begin{equation}\label{2.50}
\pi(m)(T_mC^{\perp}, T_mC^{\perp})=0.
\end{equation}
\end{definition}
\begin{definition}\label{D:2.8}
Let $(M, \pi)$ be a Poisson manifold, and $G$ a Poisson Lie group.  A left partial group action $\Phi: G\ast M\longrightarrow M$
is called a left Poisson group partial action iff
\begin{equation}\label{2.51}
\text{Graph}(\Phi) =\{(g,x,y)\in G\times M\times M\mid (g,x)\in G\ast M, y=\Phi_{g}(x)\}
\end{equation}
is a coisotropic submanifold of the product Poisson manifold $G\times M\times M^{-},$ where $M^{-}$ is
the manifold $M$ equipped with the minus Poisson structure.   In a similar way, we can define right Poisson group partial action.
\end{definition} 
We will make use of the following result.
\begin{theorem}\cite{L1}\label{T:2.9} (a)  The map $R$ in \eqref{2.41} is a Poisson diffeomorphism, when
the Poisson submanifold $\Kr\ast \Kr$ is equipped with the structure induced from $\Kr \times \Kr.$
\newline
(b) The maps
\begin{equation}\label{2.52}
\begin{aligned}
& \xi: \Kr \ast \Kr \longrightarrow \Kr, \,\,(u,v) \mapsto \xi_{u} (v)\\
& \eta: \Kr \ast \Kr\longrightarrow \Kr,\,\, (u,v) \mapsto \eta_{v} (u)\\
\end{aligned}
\end{equation}
are Poisson group partial actions.
\end{theorem}

\subsection{Dirac submanifolds, Poisson involutions, and Dirac reduction}

It is well-known that the pullback of a symplectic form to a submanifold is closed, but not necessarily nondegenerate.    In the case when the pullback is nondegenerate,
the submanifold is known as a symplectic submanifold.    In the Poisson category, there is a natural generalization of the notion of symplectic submanifolds.
For our purpose here, we will make use of the notion of Dirac submanifolds introduced in \cite{X}.
In order to define this notion, let us begin by recalling the concept of Lie algebroids.
\begin{definition}\label{D:2.10}  (a)  A Lie algebroid over a smooth manifold $M$ is a smooth vector bundle $A\longrightarrow M$ equipped with a Lie bracket $[\cdot, \cdot]$ on the set  $\Gamma(A)$ of smooth
sections of $A$ and a base-preserving bundle map $\rho: A\longrightarrow TM$ (called the
anchor map) such that 
\begin{equation}\label{2.53}
\begin{aligned}
& \rho([\xi, \eta]) = [\rho(\xi), \rho(\eta)],\\
& [\xi, f\eta] = f[\xi, \eta] + \rho(\xi)(f)\eta\\
\end{aligned}
\end{equation}
for all $\xi, \eta\in \Gamma(A)$ and for all $f\in C^{\infty}(M).$
\newline
(b)  Let $A\longrightarrow M$ be a Lie algebroid with anchor map $\rho$ and $A^{\prime}\subset A$ a vector subbundle
along a submanifold $N\subset M.$   Then $A^{\prime}\longrightarrow N$ is a Lie subalgebroid of $A$ iff the following conditions
are satisfied:
\newline
(i) if  $s_1, s_2\in \Gamma(A)$ restrict to $N$ give sections of $A^{\prime},$ then so is $[s_1\mid N, s_2\mid N].$
\newline
(ii) $\rho(A^{\prime}) \subset TN.$
\end{definition}
\begin{exmp}\label{E:2.11} (a)  A Lie algebra is a Lie algebroid over a point.
\newline
(b) Let $M$ be a smooth manifold, then the tangent bundle $TM\longrightarrow M$ is a Lie algebroid
where the Lie bracket on $\Gamma(TM)$ is the usual bracket of vector fields on $M$ and
the anchor map is the identity map $id_{TM}$ on $TM.$
\newline
(c)  Let $(M, \pi)$ be a Poisson manifold, and let $\pi^{\#}: T^*M\longrightarrow TM$ be the bundle map corresponding
to the Poisson bivector field $\pi.$  Then the \emph{cotangent Lie algebroid} \cite{Wein, F} is
the cotangent bundle $T^*M$ with anchor map given by $\pi^{\#}$ and whose space of sections $\Gamma(T^*M)$ is equipped with the Lie bracket
\begin{equation}\label{2.54}
[s_1, s_2] = L_{\pi^{\#}(s_1)} s_2 -L_{\pi^{\#}(s_2)} s_1 -d[\pi (s_1, s_2)],\,\, s_1, s_2\in \Gamma(T^*M).
\end{equation}
\end{exmp}
For a readable account on Lie algebroids including the notion of cotangent Lie algebroid, we refer the reader to \cite{CdSW}.
 \begin{definition}\label{D: 2.12}
Let $(M, \pi)$ be a Poisson manifold.   A submanifold $N$ of $M$ is a Dirac submanifold iff there exists a Whitney sum decomposition
\begin{equation}\label{2.55}
T_{N} M = TN \oplus V_{N},
\end{equation}
where $V_{N}^{\perp}$ is a Lie subalgebroid of the cotangent Lie algebroid $T^*M.$
\end{definition}

If $N$ is a Dirac submanifold of $(M, \pi),$ then necessarily $N$ carries a natural Poisson structure $\pi_{N}$ whose symplectic leaves are
given by the intersection of $N$ with the symplectic leaves of $M.$   Indeed, $\pi^{\#}_{N}: T^{*}N\longrightarrow TN$ is just the anchor
map of the Lie subalgebroid $T^{*}N\simeq V_{N}^{\perp}$  of $T^{*}M.$   Moreover, from the knowledge of the injective Lie algebroid
morphism $T^{*} N\longrightarrow T^{*}M,$ it is easy to show that \cite{X}
\begin{equation}\label{2.56}
\pi^{\#}_{N} = pr\circ \pi^{\#}\circ pr^{*},
\end{equation}
where $pr: T_{N} M\longrightarrow TN$ is the projection map induced by the decomposition in \eqref{2.55}, and $pr^{*}$ is its dual.   Note that when the
Poisson manifold is symplectic, its Dirac submanifolds are precisely its symplectic submanifolds.

The following result gives an important class of Dirac submanifolds which we will use in this work.

\begin{theorem}[\cite{X}] \label{T: 2.13}
Let  $\mu:M\longrightarrow M$ be a Poisson involution, i.e., an involution which is also a Poisson map.  Then it stable locus
$N= M^{\mu}$ is a Dirac submanifold of $M$ with $V_{N} = \bigcup_{x\in N} \text{ker}\, (T_{x} \mu +1).$
\end{theorem}
Since we will be dealing with Poisson maps between Poisson manifolds, the following result is fundamental in reducing such maps.
\begin{theorem}[\cite{L2}] \label{T:2.14}
 Let $\phi: M_1\longrightarrow M_2$ be a Poisson map and let $N_1\subset M_1,$ $N_2 \subset M_2$
be Dirac submanifolds with respective Whitney sum decompositions
\begin{equation}\label{2.57}
T_{N_1} M_1 = TN_1 \oplus V_{N_1}, \quad T_{N_2} M_2 = TN_2 \oplus V_{N_2}.
\end{equation}
Then under the assumptions that
\newline
(i) $\phi(N_1)\subset N_2,$
\newline
(ii) $T_{x}\phi (V_{N_1})_{x} \subset (V_{N_2})_{\phi(x)}, \,\, x\in N_1,$
\newline
then the reduced map $\phi\mid_{N_1} : N_1\longrightarrow N_2$ given by $\phi\mid_{N_1}(x) = \phi(x)$ for $x\in N_1$ is a Poisson map, when $N_1$ and $N_2$ are
equipped with the induced Poisson structures.
\end{theorem}
In \cite{L2}, the map $\phi\mid_{N_1} : N_1\longrightarrow N_2$ is called a {\sl Dirac reduction} of the Poisson
map $\phi: M_1\longrightarrow M_2$ and we will use this terminology here.
In the special case when the Dirac submanifolds in the theorem above are the stable loci of Poisson involutions, we
have the following result.
\begin{coro}\label{C:2.15}
Let $\mu_1: M_1\longrightarrow M_1,$ $\mu_2: M_2\longrightarrow M_2$ be Poisson involutions with
stable loci given by $N_1$ and $N_2$ respectively.  If $\phi: M_1\longrightarrow M_2$ is a Poisson map
which commutes with $\mu_1$ and $\mu_2,$ i.e., $\mu_2\circ \phi = \phi\circ \mu_1.$ then $\phi\mid_{N_1}: N_1\longrightarrow N_2$
is a Poisson map, when $N_1$ and $N_2$ are equipped with the induced structures.
\end{coro}
To end this section, we will first of all introduce a piece of notation to unify the description of reduction maps in Sections 3, 4 and 5.
For this purpose, let  $X$ be a non-empty set and consider a bijection  $\psi: X\longrightarrow X.$   Suppose  $B\subset X$ and $\psi(B)=B,$  and
let $\iota: B\longrightarrow X$  be the inclusion map.   Then the map $\psi_{\text{red}}: B\longrightarrow B$
satisfying 
\begin{equation}\label{2.58}
\psi\circ \iota = \iota\circ \psi_{\text{red}}
\end{equation}
 will be called the {\sl reduction} of $\psi$ to $B.$  Finally, we give the definition of reflection maps and parametric reflection maps.
  \begin{definition}\label{D:2.16}   
 Let $X$ be a non-empty set and suppose $\CR: X\times X\longrightarrow X\times X$
 is a Yang-Baxter map. Then $\CB: X\longrightarrow X$ is a reflection map if it satisfies the set-theoretic reflection equation
 \begin{equation}\label{2.59}
 \CB_1 \CR_{21} \CB_2 \CR_{12} = \CR_{21}\CB_2 \CR_{12} \CB_1,
 \end{equation}
 interpreted as an equality on $X\times X.$
\end{definition}
In Section 5 below, we actually have to modify Definition 2.16 a little bit, as the Yang-Baxter map there is only defined on an open submanifold $G\ast G$
of $G\times G$, and the map $B$ is only defined on some open submanifold $G^{\prime}$ of $G,$ where $G$ is a Poisson Lie group.   But the formula above does
give the correct form of the set-theoretic reflection equation. 
 \begin{definition}\label{D:2.12}
 Let $\mathcal{R}(k_1, k_2):X\times X\longrightarrow X\times X$ be a parametric Yang-Baxter map, where $k_1,k_2$ belong to
 some parameter space $\Lambda,$ and let $\rho: \Lambda\to \Lambda$ be an involution.   Then $\CB(k): X\to X,$ $k\in \Lambda,$ is called a parametric reflection 
 map if it satisfies the parametric set-theoretic reflection equation (cf. (4.13), \cite{CZ2})
 \begin{equation}\label{2.60}
 \begin{aligned} 
 & \CB_1(k_1) \CR_{21}(\rho(k_2), k_1) \CB_2(k_2) \CR_{12}(k_1,k_2)\\
 = & \CR_{21}(\rho(k_2), \rho(k_1))\CB_2(k_2)R_{12}(\rho(k_1), k_2) \CB_1(k_1),\\
 \end{aligned}
 \end{equation}
 interpreted as an equality on $X\times X.$
 \end{definition}
 As the reader will see, in \eqref{3.26} below, we  have in fact a generalization of the above form of the parametric reflection equation which 
 we call the \emph{generalized parametric reflection equation} there.  Of course, \eqref{3.26} comes from a parametric reflection equation which involves
  $R(\alpha_1,\alpha_2)$  (see \eqref{2.13}) and $B(\alpha): P(n)\longrightarrow P(n)$ where $B(\alpha)(P) = B^k(\alpha)(P)$
  for $P\in P(n)_k$ for $1\leq k\leq n-1.$   We will leave the formulation of the abstract definition of the generalized
  parametric equation to the reader.

 %%%%%%%%%%%%%%%%%%%%%%%%%%%%%%%%%%%%%%%%%%%%%%%%%%%%%%%%%%%%%%%%%%%%%%%%%%%%%%%%%%%%%
\section{Reflection maps at the level of projectors}
%%%%%%%%%%%%%%%%%%%%%%%%%%%%%%%%%%%%%%%%%%%%%%%%%%%%%%%%%%%%%%%%%% 

We begin by introducing a map
\begin{equation}\label{3.1}
U: \mathbb{C}\setminus \mathbb{R}\longrightarrow U(n)
\end{equation}
satisfying the property that
\begin{equation}\label{3.2}
U(\tau(\alpha)) = U(\alpha)^*,
\end{equation}
where $\tau$ is the involution defined by
\begin{equation}\label{3.3}
\tau:\mathbb{C}\longrightarrow \mathbb{C}: z\mapsto -\overline{z}.
\end{equation}
For $\alpha\in \mathbb{C}\setminus (\mathbb{R}\cup \sqrt{-1}\, \mathbb{R}),$ $P\in P(n),$  define a map $\sigma$ on the set of simple elements of $K_{\text{rat}}$ by
the formula
\begin{equation}\label{3.4}
\sigma(g_{\alpha, P})(z) = U(\alpha) (g^{*}_{\alpha, P}\circ\tau)(z)\, U(\alpha)^* = g_{-\overline{\alpha}, U(\alpha)PU(\alpha)^*}(z).
\end{equation}
What we would like to do is to extend $\sigma$ to a map on the entire loop group $K_{\text{rat}}$ by requiring it to be a Lie group anti-morphism, which is
the case if $U(\alpha)=I$ and we define  $\sigma(g)(z) = g^{*}(-z)$ for all $g\in K_{\text{rat}}.$   As we
will see, this is not always possible.   Since what we are doing here is motivated by the study of the $n$-Manakov system on the half-line with
$\mathbb{C}^n$-valued solution $\bm{q},$ we must
at least include the two kinds of $U(\alpha)$ which arise in \cite{CZ2}.  To recall, we have the following:
\newline
(a)  the first kind of $U(\alpha)$ is given by
\begin{equation}\label{3.5}
U(\alpha) = m(\alpha)= \frac{h(\alpha)}{|h(\alpha)|} I, \quad h(\alpha) = \frac{\alpha -i\beta}{\alpha + i\beta}, \,\, \beta\in \mathbb{R},
\end{equation}
which corresponds to imposing Robin condition on $\bm{q}$ at $x=0,$
\newline
(b) the second kind of $U(\alpha)$ consists of matrices of the form
\begin{equation}\label{3.6}
U(\alpha) = I_{S} = \text{diag}(d_1,\cdots, d_n), \,\,\text{where}\,\, d_i = \begin{cases}  1 & \,\text{if}\,\, i\in S\subset \{1,\cdots, n\},\\
-1 & \,\text{if}\,\, i\notin S,
 \end{cases}
\end{equation}
which corresponds to imposing Dirichlet condition on those components $q_i$ of the solution $\bm{q}$ with $i\in S$, and
Neumann condition on those $q_i$ with $i\notin S.$

Note, however, that in the case where $U(\alpha) = m(\alpha),$ we have $\sigma(g_{\alpha, P}) = g_{-\overline{\alpha}, P}$ and since this commutes
with $g_{\alpha, P},$ it follows that the corresponding parametric Yang-Baxter map is just the identity map.   So this case is not
interesting and for this reason, we are not going to deal with this case.    Now if we want the extension of $\sigma$ in \eqref{3.4} above to be a 
Lie-group anti-morphism, we have to make the definition
\begin{equation}\label{3.7}
\sigma(g_{\alpha_1, P_1} g_{\alpha_2, P_2}) =\sigma(g_{\alpha_2, P_2})\sigma(g_{\alpha_1, P_1}).
\end{equation}
And in order for this to be well-defined, we have to check that the result is independent of how we factorize the group
element $g = g_{\alpha_1, P_1} g_{\alpha_2, P_2}.$
So let us suppose $g_{\alpha_1, P_1} g_{\alpha_2, P_2}= g_{\alpha_2, \widetilde{P}_2} g_{\alpha_1, \widetilde{P}_1}.$ 
Then by a direct calculation using \eqref{3.4}, it is easy to check that the condition 
$\sigma(g_{\alpha_2, P_2})\sigma(g_{\alpha_1, P_1}) = \sigma(g_{\alpha_1, \widetilde{P}_1}) \sigma(g_{\alpha_2, \widetilde{P}_2})$
is not satisfied in general;  sufficient conditions which guarantee its validity are given by $U(\alpha) = m(\alpha)$ or 
$U(\alpha)$ is independent of $\alpha.$    As the former case in not of interest, we
will henceforth assume that $U$ is a constant map and by abuse of notation, we take $U$ to be a constant Hermitian
matrix in $U(n).$   It is easy to show that a Hermitian matrix $U\in U(n)$ is of the form $U = 2\Pi-I,$ where $\Pi\in P(n).$  Alternatively,
$U= V I_{S} V^*,$ where $S\subset \{1,\cdots, n\},$ and $V\in U(n).$  (This latter form of $U$ was used in \cite{CZ1} when the
authors derived integrable boundary conditions for the $n$-Manakov system.) Thus we have 
\begin{equation}\label{3.8}
\sigma(g_{\alpha, P})(z) = U (g^{*}_{\alpha, P}\circ\tau)(z)\, U^*= g_{-\overline{\alpha}, UPU^*}(z),
\end{equation}
and we can extend this to a Lie group anti-morphism of $K_{\text{rat}}$ by using the fact that $K_{\text{rat}}$ is generated by
the simple elements.  In what follows, we will denote the extension also by the same symbol $\sigma$ and we have
the general formula
\begin{equation}\label{3.9}
\sigma(g)(z) = U g^{*}(-\overline{z}) U^*, \quad g\in K_{\text{rat}}.
\end{equation}
From this, we find
\begin{equation}\label{3.10}
\begin{aligned}
\sigma^2(g) = & U(\sigma(g)^*\circ \tau) U^*\\
= & U(U((g\circ \tau)\circ \tau) U^*)U^* = g\\
\end{aligned}
\end{equation}
and therefore the map $\sigma$ is an involution.   Note that this map $\sigma$ is not used to impose a reduction in the sense of \cite{Mi}; indeed, $\sigma$
is a Lie group anti-morphism and therefore the fixed point set of $\sigma$ is not a subgroup.    In the next proposition, the reader will see the role played
by this map.
\begin{prop}\label{P: 3.1}
Given $\alpha\in \mathbb{C}\setminus (\mathbb{R}\cup \sqrt{-1}\,\mathbb{R}),$ $P\in P(n)_k,$ the refactorization problem
\begin{equation}\label{3.11}
\sigma(g_{\alpha, P})g_{\alpha, P} \,\ = g_{\alpha, \widetilde{P}_2} g_{-\overline{\alpha}, \widetilde{P}_1} 
\end{equation}
has a unique solution. Moreover, we have
\begin{equation}\label{3.12}
\widetilde{P}_2 = \phi P \phi^{-1}, \,\, \phi = 2\alpha I + (\overline{\alpha}-\alpha)(P + UPU^*),
\end{equation}
and
\begin{equation}\label{3.13}
\widetilde{P}_1 = U \widetilde{P}_2 U^*.
\end{equation}
\end{prop}
\begin{proof}
First of all, since $\alpha$ is neither on the real line nor on the imaginary axis, it follows that $\{\alpha, \overline{\alpha}\}\cap \{-\alpha, -\overline{\alpha}\}=\emptyset.$
Therefore, it follows from Theorem 2.2 that the refactorization problem is guaranteed to have a unique solution for $(\widetilde{P}_1, \widetilde{P}_2).$   Moreover,
the formulas for $\widetilde{P}_2$ and $\phi$ in \eqref{3.12} follow immediately from \eqref{2.11} and \eqref{2.12}.

Now, observe that
\begin{equation}\label{3.14}
\phi U = 2\alpha U + (\overline{\alpha}-\alpha) (PU + UP) = U\phi
\end{equation}
and therefore we also have $U\phi^{-1} = \phi^{-1}U.$   From this, we find that
\begin{equation}\label{3.15}
\widetilde{P}_1 = \phi UPU^* \phi^{-1} = U\phi P\phi^{-1} U^* = U \widetilde{P}_2 U^*,
\end{equation}
as asserted.    Note that the same conclusion can also be obtained from the invariance of the left hand side of \eqref{3.11} under $\sigma$
from which it follows that 
\begin{equation} \label{3.16}
g_{\alpha, \widetilde{P}_2} g_{-\overline{\alpha}, \widetilde{P}_1} = \sigma(g_{-\overline{\alpha}, \widetilde{P}_1})\sigma(g_{{\alpha}, \widetilde{P}_2}).
\end{equation}
Thus it follows from the uniqueness of solution of the refactorization problem (Theorem 2.2) that
$g_{-\overline{\alpha}, \widetilde{P}_1} = \sigma(g_{\alpha, \widetilde{P}_2})$ and this gives \eqref{3.13}.
\end{proof}

In view of the special relation between $P_1 = U P U^*$ and $P_2 = P$ in the above refactorization problem in \eqref{3.11}, we introduce the map
\begin{equation}\label{3.17}
c_{U,k} : P(n)_k \longrightarrow P(n)_k,  P\mapsto U P U^*
\end{equation}
and we would like to restrict $R^{k,k}(\tau(\alpha), \alpha)$ (see \eqref{2.18}) to the graph of $c_{U,k},$ which we define to be 
\begin{equation}\label{3.18}
 \mathcal{G}_{U,k}=\{ (c_{U,k}(P), P)\mid P\in P(n)_k \}.
 \end{equation}
 Indeed, if
we let $\iota_{k}: \mathcal{G}_{U,k}\longrightarrow P(n)_k\times P(n)_k$ be the inclusion map, then Proposition 3.1 shows that
we have the reduced diffeomorphism
\begin{equation}\label{3.19}
\begin{aligned}
R^{k,k}_{\text{red}}( \tau(\alpha), \alpha) : \mathcal{G}_{U,k}\longrightarrow \mathcal{G}_{U,k}: ( c_{U,k}(P), P)\mapsto ( c_{U,k}(\widetilde{P}_2), \widetilde{P}_2),
\end{aligned}
\end{equation}
satisfying the relation  $\iota_{k}\circ R^{k,k}_{\text{red}}(\tau(\alpha), \alpha) = R^{k,k}(\tau(\alpha), \alpha)\circ \iota_{k},$  where $\widetilde{P}_2$ is
given by \eqref{3.12}.   Now we introduce
the pre-symplectic form
\begin{equation}\label{3.20}
\omega^k = \iota_{k}^{*}(\omega_{E_k} \oplus \omega_{E_k})
\end{equation}
on $\mathcal{G}_{U,k},$ where we recall that $\omega_{E_k}$ is defined in \eqref{2.8}.
Note that due to the relation $\text{Im}\, \alpha = \text{Im}\, (-\overline{\alpha}),$  we can simply drop the common factor $-2\text{Im}\, \alpha$
in the expression for $\omega_{-\overline{\alpha}, \alpha}$ and consider $\omega_{E_k} \oplus \omega_{E_k}$ in \eqref{3.20} above.
In our next result, we will show that $\omega^k$ is nondegenerate on $\mathcal{G}_{U,k}.$
\begin{prop}\label{P: 3.2}
The $2$-form $\omega^k$ is nondegenerate on $\mathcal{G}_{U,k}$ so that 
$(\mathcal{G}_{U,k}, \omega^k)$ is a symplectic submanifold of  $(P(n)_k\times P(n)_k, \omega_{E_k} \oplus \omega_{E_k}).$
Hence  $R^{k,k}_{\text{red}}(\tau(\alpha), \alpha)$ is a symplectomorphism when the domain and codomain are equipped with the symplectic
form $\omega^k .$
\end{prop}
\begin{proof}
Suppose $(c_{U,k}(H), H) \in \text{Ker}\, (\omega^k)_{(c_{U,k}(P), P)},$ then $H = [X, P]$ for some
$X\in u(n).$   Therefore, we have
\begin{equation}\label{3.21}
(\omega^k)_{(c_{U,k}(P), P)} \Big((c_{U,k}([X, P]), [X,P]), (c_{U,k}([Y,P]), [Y,P])\Big)=0
\end{equation}
for all $Y\in u(n).$
But by a direct computation using the definition of $\omega^k,$ we have
\begin{equation}\label{3.22}
\begin{aligned}
& (\omega^k)_{(c_{U,k}(P), P)} \Big( (c_{U,k}([X, P]), [X, P]), (c_{U,k}([Y,P]), [Y, P])\Big)\\
= & \omega_{E_k}(c_{U,k}(P)) ([c_{U,k}(X), c_{U}(P)], [c_{U,k}(Y), c_{U,k}(P)]) + \omega_{E_k} (P)([X, P], [Y, P])\\
= & \sqrt{-1}\,\text{tr}\, c_{U,k}(P) [ c_{U,k}(X), c_{U,k}(Y)] +\sqrt{-1} \,\text{tr}\, P[X, Y] \\
= & -2\sqrt{-1}\, \text{tr}\, [X, P] Y.
\end{aligned}
\end{equation}
Since this expression is zero for all $Y\in u(n),$ it follows that we must have $H = [X, P]=0.$   As this is true for all points $(P, c_{U,k}(P))\in \mathcal{G}_{U,k},$
this shows the closed $2$-form $\omega^k$ is nondegenerate.   The assertion that  $R^{k,k}_{\text{red}}( \tau(\alpha), \alpha)$ is a symplectomorphism
then follows from the relation $\iota_{k}\circ R^{k, k}_{\text{red}}( \tau(\alpha), \alpha) = R^{k,k}( \tau(\alpha), \alpha)\circ \iota_{k}$ and Theorem 2.2 (c).
\end{proof}
We will call $R^{k,k}_{\text{red}}(\tau(\alpha), \alpha)$  the {\sl reduced parametric Yang-Baxter map}.
\begin{definition}\label{D: 3.3}
 We define the map $B^k$ by
\begin{equation}\label{3.23}
\begin{aligned}
B^k: & (\mathbb{C}\setminus (\mathbb{R}\cup \sqrt{-1}\, \mathbb{R}))\times P(n)_k\longrightarrow (\mathbb{C}\setminus (\mathbb{R}\cup \sqrt{-1}\,
 \mathbb{R}))\times P(n)_k,\\ 
& (\alpha, P)\mapsto (\tau(\alpha), B^k(\alpha)(P))=(\tau(\alpha), U \phi P \phi^{-1} U^*),\\
\end{aligned}
\end{equation}
where $\phi$ is given in \eqref{3.12}.  
\end{definition}
From the above definitions and \eqref{3.19}, we clearly have
\begin{equation}\label{3.24}
R^{k,k}_{\text{red}}(\tau(\alpha), \alpha) (c_{U, k}(P), P) = (B^k(\alpha)(P), c_{U,k}(B^k(\alpha)(P))).
\end{equation}
\begin{theorem}\label{T: 3.4}
(a)  $B^k$ is an involution.   In particular, we have
\begin{equation}\label{3.25}
B^k(\tau(\alpha)) B^k(\alpha) = \text{id}_{P(n)_k}.
\end{equation}
\newline
(b)  For any $1\leq k, \ell\leq n-1,$ the pair $B^k(\alpha)$ and $B^{\ell}(\alpha)$ satisfies the generalized parametric reflection equation
\begin{equation}\label{3.26}
\begin{aligned}
& B_1^k(\alpha_1) R_{21}^{k,\ell}(\tau(\alpha_2), \alpha_1)  B_2^{\ell}(\alpha_2)R_{12}^{k,\ell}(\alpha_1,\alpha_2) \\
 = & R_{21}^{k,\ell}(\tau(\alpha_2), \tau(\alpha_1))B^{\ell}_2(\alpha_2)R_{12}^{k,\ell}(\tau(\alpha_1), \alpha_2) B_1^{k}(\alpha_1)\\
\end{aligned}
\end{equation}
for all $\alpha_1, \alpha_2\in \mathbb{C}\setminus (\mathbb{R}\cup \sqrt{-1}\, \mathbb{R})$ satisfying the conditions
\begin{equation}\label{3.27}
\{\alpha_1, \overline{\alpha}_1\}\cap \{\alpha_2, \overline{\alpha}_2\}=\emptyset,\,\, \{-\alpha_1, -\overline{\alpha}_1\}\cap \{\alpha_2, \overline{\alpha}_2\}=\emptyset.
\end{equation}
When $k=\ell,$ the generalized parametric reflection equation reduces to the usual parametric equation \eqref{2.60}.
\newline
(c)  The map
\begin{equation}\label{3.28}
R^k({\alpha}): (P(n)_k, 2\omega_{E_k})\longrightarrow  (\mathcal{G}_{U,k}, \omega^k): P\mapsto R^{k,k}_{\text{red}}(\tau(\alpha), \alpha) (c_{U,k}(P), P) 
\end{equation}
is a symplectomorphism.   Moreover, the parametric reflection map 
\begin{equation}\label{3.29}
\begin{aligned}
B^k(\alpha) & = (id_{P(n)_k}, c_{U,k})^{-1} \circ R^k(\alpha)\\
& = (c_{U,k}, id_{P(n)_k})^{-1}\circ
 (S\circ R^{k,k}_{\text{red}}(\tau(\alpha), \alpha))\circ (c_{U,k}, id_{P(n)_k})\\
 \end{aligned}
 \end{equation}
 is also a symplectomorphism, when $P(n)_k$ is equipped with the symplectic form $\omega_{E_k}.$  Here
 \begin{equation}\label{3.30}
(c_{U,k}, id_{P(n)_k}): P(n)_k \longrightarrow \mathcal{G}_{U,k}, P\mapsto  (c_{U,k}(P), P)
\end{equation}
and $S: \mathcal{G}_{U,k}\longrightarrow \mathcal{G}_{U,k}$ is the restriction of the permutation map on
$P(n)_k \times P(n)_k$ that sends $(P_1, P_2)$ to $(P_2, P_1).$
 \end{theorem}
\begin{proof}
(a)  To simplify notation, we denote $B^k$ simply by $B.$ To show $B^2(\alpha, P) = (\alpha, P),$ we make use the the uniqueness of solution of refactorization
problems.   In what follows, in order to facilitate our calculations, we will (by abuse of notation) denote
$g_{\tau(\alpha),  U\phi P\phi^{-1} U^*}$  more economically as $g_{B(\alpha, P)}$ (see \eqref{3.23} above).  With this notation, the refactorization problem which defines $B(\alpha, P)$ is given by
\begin{equation}\label{3.31}
\sigma(g_{\alpha, P})g_{\alpha, P} =  \sigma(g_{B(\alpha, P)}) g_{B(\alpha, P)}.
\end{equation}
Similarly, the refactorization which defines $B^2(\alpha, P) = B(B(\alpha,P))$ is given by
\begin{equation}\label{3.32}
 \sigma(g_{B(\alpha, P)}) g_{B(\alpha, P)} = \sigma(g_{B^2(\alpha, P)})g_{B^2(\alpha, P)}.
\end{equation}
Hence it follows from the last two expressions and the uniqueness of solutions of refactorization problems that $B^2$ is the
identity map.  As $B(B(\alpha, P)) = B(\tau(\alpha), B(\alpha)(P)) = (\alpha, B(\tau(\alpha))B(\alpha)(P)),$ the second
assertion follows.
\newline
(b)  Take $\alpha_1,\alpha_2$ satisfying the assumptions, and let $\alpha_3 = \tau(\alpha_1),$ $\alpha_4=\tau(\alpha_2).$   We consider the graph of 
$c_{U,k}\times c_{U,\ell},$ which we define to be
\begin{equation}\label{3.33}
\mathcal{G}(U,k,\ell)=\{(P_1, P_2, c_{U,k}(P_1), c_{U,\ell}(P_2)\mid (P_1,P_2)\in P(n)_k\times P(n)_{\ell}\}.
\end{equation}
In the next two expressions, we consider maps $R_{ij}(\alpha_i,\alpha_j)$ apply to quadruples of projectors $(P_1, P_2, P_3, P_4)\in P(n)_{k_1}\times P(n)_{k_2}\times
P(n)_{k_3} \times P(n)_{k_4}.$  The notation means the following: if $i<j$ (resp.~ $i>j$),  $R_{ij}(\alpha_i,\alpha_j)$ is the map which acts as $R(\alpha_i,\alpha_j)$ 
(resp.~ $R_{21}(\alpha_i,\alpha_j)$) on the $i$-th and $j$-the component and as identity on the other components (see Theorem 2.2 for the definition of $R(\alpha_i,\alpha_j)$ and \eqref{2.19} for the definition of $R_{21}(\alpha_i,\alpha_j)$). 
Let 
\begin{equation}\label{3.34}
\begin{aligned}
& \Pi_{1} (\alpha_1,\alpha_2, \alpha_3, \alpha_4)\\
= & R_{31}^{k_1,k_3}(\alpha_3,\alpha_1) R_{32}^{k_2,k_3}(\alpha_3,\alpha_2) R_{41}^{k_1,k_4}(\alpha_4,\alpha_1) R_{42}^{k_2,k_4} (\alpha_4,\alpha_2) R_{43}^{k_3,k_4}(\alpha_4,\alpha_3) R_{12}^{k_1,k_2}(\alpha_1,\alpha_2),
\end{aligned}
\end{equation}
and 
\begin{equation}\label{3.35}
\begin{aligned}
& \Pi_{2} (\alpha_1,\alpha_2, \alpha_3, \alpha_4)\\
= & R_{43}^{k_3,k_4}(\alpha_4,\alpha_3) R_{12}^{k_1,k_2}(\alpha_1,\alpha_2) R_{42}^{k_2,k_4}(\alpha_4,\alpha_2) R_{32}^{k_2,k_3}(\alpha_3,\alpha_2) R_{41}^{k_1,k_4}(\alpha_4,\alpha_1) R_{31}^{k_1,k_3}(\alpha_3,\alpha_1),\\
\end{aligned}
\end{equation}
where we assume the parameters $\alpha_a$ are such that all the maps on the right hand sides of the above two expressions are defined.
In the following calculation, we will drop the superscripts to simplify notation.  Using the Yang-Baxter equation in the form
\begin{equation}\label{3.36}
R_{ab}(\alpha_a, \alpha_b) R_{ac}(\alpha_a, \alpha_c)R_{bc}(\alpha_b, \alpha_c) =R_{bc}(\alpha_b, \alpha_c)R_{ac}(\alpha_a, \alpha_c)R_{ab}(\alpha_a, \alpha_b)
\end{equation}
valid for any triplet $a,b,c\in \{1,2,3,4\}$ and the fact that
\begin{equation}\label{3.37}
R_{ab}(\alpha_a,\alpha_b)R_{cd}(\alpha_c,\alpha_d) = R_{cd}(\alpha_c,\alpha_d) R_{ab}(\alpha_a,\alpha_b)
\end{equation}
for any pairwise distinct $a, b, c, d,$  we have
\begin{equation}\label{3.38}
\begin{aligned}
& \Pi_1(\alpha_1,\alpha_2,\alpha_3,\alpha_4)\\
= & R_{31}(\alpha_3,\alpha_1) R_{32}(\alpha_3,\alpha_1) R_{41}(\alpha_4,\alpha_1) R_{42} (\alpha_4,\alpha_2)  R_{12}(\alpha_1,\alpha_2)  R_{43}(\alpha_4,\alpha_3) \\
= & R_{31}(\alpha_3,\alpha_1)R_{32}(\alpha_3,\alpha_1) R_{12}(\alpha_1,\alpha_2) R_{42} (\alpha_4,\alpha_2)R_{41}(\alpha_4,\alpha_1)
R_{43}(\alpha_4,\alpha_1)\\
= & R_{12}(\alpha_1,\alpha_2)R_{32}(\alpha_3,\alpha_1)R_{31}(\alpha_3,\alpha_1)R_{42} (\alpha_4,\alpha_2)R_{41}(\alpha_4,\alpha_1)
R_{43}(\alpha_4,\alpha_1)\\
= & R_{12}(\alpha_1,\alpha_2)R_{32}(\alpha_3,\alpha_1)R_{42} (\alpha_4,\alpha_2)R_{31}(\alpha_3,\alpha_1)R_{41}(\alpha_4,\alpha_1)
R_{43}(\alpha_4,\alpha_1)\\
= & R_{12}(\alpha_1,\alpha_2)R_{32}(\alpha_3,\alpha_1)R_{42} (\alpha_4,\alpha_2)R_{43}(\alpha_4,\alpha_1)R_{41}(\alpha_4,\alpha_1)
R_{31}(\alpha_3,\alpha_1)\\
= & R_{12}(\alpha_1,\alpha_2)R_{43}(\alpha_4,\alpha_1))R_{42} (\alpha_4,\alpha_2)R_{32}(\alpha_3,\alpha_1)R_{41}(\alpha_4,\alpha_1)
R_{31}(\alpha_3,\alpha_1)\\
= & \Pi_2(\alpha_1,\alpha_2,\alpha_3,\alpha_4).\\
\end{aligned}
\end{equation}
Therefore, in particular, for $\alpha_1,\alpha_2$ satisfying \eqref{3.27} and $\alpha_3 = \tau(\alpha_1),$  $\alpha_4 = \tau(\alpha_2),$  we obtain
\begin{equation}\label{3.39}
\Pi_1(\alpha_1,\alpha_2,\alpha_3,\alpha_4)\mid_{\mathcal{G}(U,k,\ell)}= \Pi_2(\alpha_1,\alpha_2,\alpha_3,\alpha_4)\mid_{\mathcal{G}(U,k,\ell)}.
\end{equation}
Now consider the product $\sigma(g_{\alpha_1, P_1} g_{\alpha_2, P_2}) g_{\alpha_1, P_1} g_{\alpha_2, P_2},$  where $(P_1, P_2)$
 belongs to $P(n)_k\times P(n)_{\ell}.$  Then corresponding to $\Pi_{1} (\alpha_1,\alpha_2, \alpha_3, \alpha_4)\mid_{\mathcal{G}(U,k,\ell)}$ applied to $(P_1, P_2, c_{U,k}(P_1), c_{U,k}(P_2)),$ we have the following sequence of expressions, obtained by solving successively the respective refactorization problems, which have unique solutions by the assumptions on $\alpha_1,$ $\alpha_2$:
\begin{equation}\label{3.40}
\begin{aligned}
& \sigma(g_{\alpha_2, P_2})\sigma(g_{\alpha_1, P_1}) g_{\alpha_1, P_1} g_{\alpha_2, P_2} \\
= &\sigma(g_{\alpha_1, P_1^{(ii)}})\sigma(g_{\alpha_2, P_2^{(ii)}}) g_{\alpha_2, P_2^{(ii)}} g_{\alpha_1, P_1^{(ii)}}\\
= & \sigma(g_{\alpha_1, P_1^{(ii)}}) g_{\alpha_2, P_2^{(iii)}} \sigma(g_{\alpha_2, P_2^{(iii)}}) g_{\alpha_1, P_1^{(ii)}}\\ 
= & g_{\alpha_2, P_2^{(iv)}} \sigma(g_{\alpha_1, P_1^{(iv)}})g_{\alpha_1, P_1^{(iv)}} \sigma(g_{\alpha_2, P_2^{(iv)}})\\
= & g_{\alpha_2, P_2^{(iv)}} g_{\alpha_1, P_1^{(v)}}\sigma(g_{\alpha_1, P_1^{(v)}}) \sigma(g_{\alpha_2,
P_2^{(iv)}}) ,\\
\end{aligned}
\end{equation}
so that
\begin{equation}\label{3.41}
\begin{aligned}
& \Pi_1(\alpha_1, \alpha_2, \alpha_3, \alpha_4) (P_1, P_2, c_{U,k}(P_1), c_{U,\ell}(P_2))\\
= & (P_1^{(v)}, P_2^{(iv)}, c_{U,k}(P_1^{(v)}), c_{U,\ell}(P_2^{(iv)})).\\
\end{aligned}
\end{equation}
In a similar way, corresponding to $\Pi_2(\alpha_1, \alpha_2, \alpha_3, \alpha_4)\mid_{\mathcal{G}(U,k,\ell)}$ applied to the same quadruple in $\mathcal{G}(U,k,\ell),$
we find
\begin{equation}\label{3.42}
\begin{aligned}
& \sigma(g_{\alpha_2, P_2})\sigma(g_{\alpha_1, P_1}) g_{\alpha_1, P_1} g_{\alpha_2, P_2} \\
= & \sigma(g_{\alpha_2, P_2}) g_{\alpha_1, \widetilde{P}_1^{(ii)}}\sigma(g_{\alpha_1, \widetilde{P}_1^{(ii)}}) g_{\alpha_2, P_2} \\
= & g_{\alpha_1, \widetilde{P}_1^{(iii)}} \sigma(g_{\alpha_2, \widetilde{P}_2^{(iii)}})g_{\alpha_2, \widetilde{P}_2^{(iii)}}
\sigma(g_{\alpha_1, \widetilde{P}_1^{(iii)}})\\
= & g_{\alpha_1, \widetilde{P}_1^{(iii)}} g_{\alpha_2, \widetilde{P}_2^{(iv)}}\sigma(g_{\alpha_2, \widetilde{P}_2^{(iv)}})
\sigma(g_{\alpha_1, \widetilde{P}_1^{(iii)}})\\
= &  g_{\alpha_2, \widetilde{P}_2^{(v)}}g_{\alpha_1, \widetilde{P}_1^{(v)}} \sigma(g_{\alpha_1, \widetilde{P}_1^{(v)}})
\sigma(g_{\alpha_2, \widetilde{P}_2^{(v)}}).\\
\end{aligned}
\end{equation}
which means that 
\begin{equation}\label{3.43}
\begin{aligned}
& \Pi_2(\alpha_1, \alpha_2, \alpha_3, \alpha_4) (P_1, P_2, c_{U,k}(P_1), c_{U,k}(P_2))\\
= & (\widetilde{P}_1^{(v)}, \widetilde{P}_2^{(v)}, c_{U,k}(\widetilde{P}_1^{(v)}), c_{U,\ell}(\widetilde{P}_2^{(v)})).\\
\end{aligned}
\end{equation}
Equating \eqref{3.41} and \eqref{3.43}, we conclude that
\begin{equation}\label{3.44}
(P_1^{(v)}, P_2^{(iv)}, c_{U,k}(P_1^{(v)}), c_{U,\ell}(P_2^{(iv)}))=(\widetilde{P}_1^{(v)}, \widetilde{P}_2^{(v)}, c_{U,k}(\widetilde{P}_1^{(v)}), c_{U,\ell}(\widetilde{P}_2^{(v)})).
\end{equation}
To conclude the proof, we make use of \eqref{3.40} again and obtain
\begin{equation}\label{3.45}
\begin{aligned}
& B_1(\alpha_1) R_{21}(\tau(\alpha_2), \alpha_1)  B_2(\alpha_2)R_{12}(\alpha_1,\alpha_2)(P_1, P_2)\\
= & B_1(\alpha_1) R_{21}(\tau(\alpha_2), \alpha_1)  B_2(\alpha_2) (P_1^{(ii)}, P_2^{(ii)})\\
= & B_1(\alpha_1) R_{21}(\tau(\alpha_2), \alpha_1) (P_1^{(ii)}, c_{U,\ell} (P_2^{(iii)}))\\
= & B_1(\alpha_1) (P_1^{(iv)}, c_{U,\ell}(P_2^{(iv)}))\\
= & (c_{U,k}(P_1^{(v)}), c_{U,\ell}(P_2^{(iv)})).\\
\end{aligned}
\end{equation}
Similarly, by using \eqref{3.42} again, we find
\begin{equation}\label{3.46}
\begin{aligned}
& R_{21}(\tau(\alpha_2), \tau(\alpha_1))B_2(\alpha_2)R_{12}(\tau(\alpha_1), \alpha_2) B_1(\alpha_1)(P_1, P_2)\\
= &  R_{21}(\tau(\alpha_2), \tau(\alpha_1))B_2(\alpha_2)R_{12}(\tau(\alpha_1), \alpha_2)( c_{U,k}(\widetilde{P}_1^{(ii)}), P_2)\\
= & R_{21}(\tau(\alpha_2), \tau(\alpha_1))B_2(\alpha_2)(c_{U,k}(\widetilde{P}_1^{(iii)}), \widetilde{P}_2^{(iii)})\\
= & R_{21}(\tau(\alpha_2), \tau(\alpha_1)) (c_{U,k}(\widetilde{P}_1^{(iii)}),  c_{U,\ell}(\widetilde{P}_2^{(iv)}))\\
= & (c_{U,k}(\widetilde{P}_1^{(v)}), c_{U,\ell}(\widetilde{P}_2^{(v)})).\\
\end{aligned}
\end{equation}
Hence the equality of \eqref{3.45} and \eqref{3.46} follows from \eqref{3.44}.
\newline
(c)  With the definition of the map $(c_{U,k}, id_{P(n)_k})$ in \eqref{3.30}, it is clear that we can regard $R^k(\alpha)$ 
 as the the composite
\begin{equation}\label{3.47}
R^k(\alpha) = R^{k,k}_{\text{red}}(\tau(\alpha), \alpha) \circ (c_{U,k}, id_{P(n)_k}).
\end{equation}
Now $c_{U,k}$ is a symplectomorphism, when $P(n)_k$ is equipped with the $2$-form $\omega_{E_k}.$
By direct calculation, we have
\begin{equation}\label{3.48}
 ( c_{U,k}, id_{P(n)_k})^* \omega^k
= c_{U,k}^* (\omega_{E_k}) +  \omega^{\prime}_{E_k} = 2\omega_{E_k} ,
\end{equation}
hence the map 
\begin{equation}\label{3.49}
(c_{U,k},  id_{P(n)_k}): (P(n)_k, 2\omega_{E_k})\longrightarrow (\mathcal{G}_{U,k}, \omega^k)
\end{equation}
 is  a symplectomorphism as well.   As the composition of symplectomorphisms is a symplectomorphism, it follows from the above
 argument, \eqref{3.48}, and 
 Theorem 3.4 (c)  that $R^k(\alpha)$ is a symplectomorphism.    To establish the corresponding assertion for $B(\alpha),$
 first note that
 \begin{equation}\label{3.50}
 B^k(\alpha) = \pi_1\circ \iota_{k}\circ R^k(\alpha),
 \end{equation}
 where $\iota_{k}$ is the inclusion map of $\mathcal{G}_{U,k}$ in $P(n)_k\times P(n)_k,$ and
 $\pi_1$ is the projection map of this product space into the first factor.    But now it is easy to show that
 $\pi_1 \circ \iota_{k} =  (id_{P(n)_k}, c_{U,k})^{-1} = (c_{U,k}, id_{P(n)_k})^{-1}\circ S.$    Hence the assertion for $B^k(\alpha)$ follows as we can drop the factor $2$ from the $2$-form $2\omega_{E_k},$ and $S$ is clearly a symplectomorphism.
\end{proof}
In view of Theorem 3.4  and Definition 2.17, $B^k(\alpha)$ is a  {\sl parametric reflection map}.

%%%%%%%%%%%%%%%%%%%%%%%%%%%%%%%%%%%%%%%%%%%%%%%%%%%%%%%%%%%%%%%%%%%%%%%%%%%%%%%%%%%%
\section{The rank $1$ case}
%%%%%%%%%%%%%%%%%%%%%%%%%%%%%%%%%%%%%%%%%%%%%%%%%%%%%%%%%%%%%%%%%%%%%%%%%%%%%%%%%%%

The case where the projectors are of rank $1$ is of special interest because it arises in soliton-boundary interactions in multi-component soliton equations such
as the $n$-Manakov system on the half-line \cite{CZ2}.   Of course, what we are doing here is beyond what is required to understand the results in \cite{CZ2}.

We begin by introducing quantities related to the ones which appear in the last section in the case when the projectors are of rank $1.$
First, we have the diffeomorphism
\begin{equation}\label{4.1}
\widetilde{c}_{U}:  \mathbb{C}\mathbb{P}^{n-1}\longrightarrow \mathbb{C}\mathbb{P}^{n-1}, [p]\mapsto [Up]
\end{equation}
which is conjugate to $c_{U,1}$ with the relation 
\begin{equation}\label{4.2}
\widetilde{c}_{U} = j_{\delta}^{-1} \circ c_{U,1} \circ j_{\delta},
\end{equation}
where $j_{\delta}$ is defined in \eqref{2.22}
Corresponding to the map $\widetilde{c}_{U}$ is its graph, which we define to be the submanifold of $\mathbb{C}\mathbb{P}^{n-1}\times \mathbb{C}\mathbb{P}^{n-1},$
given by
\begin{equation}\label{4.3}
\widetilde{\mathcal{G}}_{U}= \{([Up], [p])\mid [p]\in \mathbb{C}\mathbb{P}^{n-1}\}.
\end{equation}
Let $\widetilde{\iota}_1: \widetilde{\mathcal{G}}_{U}\longrightarrow \mathbb{C}\mathbb{P}^{n-1}\times \mathbb{C}\mathbb{P}^{n-1}$ be the inclusion map, 
and let
\begin{equation}\label{4.4}
(j_{\delta}\times j_{\delta})_{r}: \widetilde{\mathcal{G}}_{U}\longrightarrow \mathcal{G}_{U,1}, \,\, ([Up], [p])\mapsto (\pi_{[Up]}, \pi_{[p]}) = (U\pi_{[p]}U^{*}, \pi_{[p]})
\end{equation}
be the map induced by $j_{\delta}\times j_{\delta}$,
then by conjugating $R^{1,1}_{\text{red}}(\tau(\alpha),\alpha)$ by the map $(j_{\delta}\times j_{\delta})_r^{-1},$
we have the reduced diffeomorphism
\begin{equation}\label{4.5}
 \widetilde{R}_{\text{red}}(\tau(\alpha), \alpha): \widetilde{\mathcal{G}}_{U}\longrightarrow \widetilde{\mathcal{G}}_{U},
\end{equation}
where $R^{k,k}_{\text{red}}(\tau(\alpha), \alpha)$ is defined in \eqref{3.19}.  Clearly, $\widetilde{R}_{\text{red}}(\tau(\alpha), \alpha)$ is the
reduction of $\widetilde{R}(\tau(\alpha),\alpha):= (j_{\delta}\times j_{\delta})^{-1}\circ R^{1,1}(\tau(\alpha), \alpha)\circ (j_{\delta}\times j_{\delta})$ to
$\widetilde{\mathcal{G}}_{U}.$  We now define the pre-symplectic form
\begin{equation}\label{4.6}
\widetilde{\omega} = \widetilde{\iota}_1^{*}(\omega_{FS} \oplus \omega_{FS}),
\end{equation}
where $\omega_{FS}$ is the Fubini-Study $2$-form in \eqref{2.23}.
The proof of the next proposition makes use of the invariance of $\omega_{FS}$ under
$U(n)$, Theorem 2.3 (b), and the relation $\widetilde{\iota}_1\circ \widetilde{R}_{\text{red}}(\tau(\alpha), \alpha) = \widetilde{R}(\tau(\alpha), \alpha)\circ \widetilde{\iota}_1,$  it proceeds in a similar way as in the proof of Proposition 3.2.
\begin{prop}\label{P: 4.1}
The $2$-form $\widetilde{\omega}$ is nondegenerate on $\widetilde{\mathcal{G}}_{U}$ so that $(\widetilde{\mathcal{G}}_{U}, \widetilde{\omega})$
is a symplectic submanifold of $(\mathbb{C}\mathbb{P}^{n-1}\times \mathbb{C}\mathbb{P}^{n-1}, \omega_{FS}\oplus \omega_{FS})$.
Hence $\widetilde{R}_{\text{red}}(\tau(\alpha) \alpha)$ is a symplectomorphism when its domain and codomain
are equipped with the symplectic form $\widetilde{\omega}.$
\end{prop} 
We next introduce the analog of Definition 3.3, for the case $k=1,$ at the level of complex projective space.
\begin{definition}\label{D: 4.2}
We define the  map  
\begin{equation}\label{4.7}
\widetilde{B} : (\mathbb{C}\setminus (\mathbb{R}\cup \sqrt{-1}\, \mathbb{R}))\times \mathbb{C}\mathbb{P}^{n-1}\longrightarrow 
(\mathbb{C}\setminus (\mathbb{R}\cup \sqrt{-1}\, \mathbb{R}))\times \mathbb{C}\mathbb{P}^{n-1}
\end{equation}
by
\begin{equation}\label{4.8}
\widetilde{B} = (id \times j_{\delta})^{-1}\circ B^1 \circ (id \times j_{\delta}),
\end{equation}
where $id$ is the identity map of the parameter space $\mathbb{C}\setminus (\mathbb{R}\cup \sqrt{-1}\, \mathbb{R})$
and $B$ is defined in \eqref{3.23}.
The corresponding parametric  map is defined by the relation
\begin{equation}\label{4.9}
\widetilde{B}(\alpha, [p]) = (\tau(\alpha), \widetilde{B}(\alpha)([p])).
\end{equation}
\end{definition}
From the above definition, we have (cf. \eqref{3.24})
\begin{equation}\label{4.10}
\widetilde{R}(\tau(\alpha), \alpha)(\widetilde{c}_{U}([p]), [p]) = (\widetilde{B}(\alpha)([p]), \widetilde{c}_{U}(\widetilde{B}(\alpha)([p])).
\end{equation}
The next result is the analog of Theorem 3.4 at the level of $\mathbb{C}\mathbb{P}^{n-1},$ for the case where $k=1.$
\begin{prop}\label{P: 4.3}
 (a)  The map $\widetilde{B}$ is an involution.   In particular, $\widetilde{B}$ satisfies
\begin{equation}\label{4.11}
\widetilde{B}(\tau(\alpha))\widetilde{B}(\alpha)=\text{id}_{\mathbb{C}\mathbb{P}^{n-1}}.
\end{equation}
 Explicitly,
\begin{equation}\label{4.12}
\begin{aligned}
\widetilde{B}(\alpha)([p]) = & [ (\alpha + \overline{\alpha}) g_{\overline{\alpha}, \pi_{[p]}}(\tau(\alpha)) Up]\\
= & \left[ \left(I + \frac{\overline{\alpha} -\alpha}{\alpha + \overline{\alpha}} \frac{pp^*}{p^*p}\right) Up \right].
\end{aligned}
\end{equation}
(b)  $\widetilde{B}(\alpha)$ satisfies the parametric reflection equation
\begin{equation}\label{4.13}
\begin{aligned}
& \widetilde{B}_1(\alpha_1) \widetilde{R}_{21}(\tau(\alpha_2), \alpha_1)  \widetilde{B}_2(\alpha_2)\widetilde{R}_{12}(\alpha_1,\alpha_2) \\
 = & \widetilde{R}_{21}(\tau(\alpha_2), \tau(\alpha_1))\widetilde{B}_2(\alpha_2)\widetilde{R}_{12}(\tau(\alpha_1), \alpha_2) \widetilde{B}_1(\alpha_1)\\
\end{aligned}
\end{equation}
for all $\alpha_1, \alpha_2\in \mathbb{C}\setminus (\mathbb{R}\cup \sqrt{-1}\, \mathbb{R})$ satisfying the conditions in \eqref{3.27}.
\newline
(c)  The map
\begin{equation}\label{4.14}
\widetilde{R}(\alpha): (\mathbb{C}\mathbb{P}^{n-1}, 2 \omega_{FS}) \longrightarrow (\widetilde{\mathcal{G}}_{U}, \widetilde{\omega}): [p]\mapsto \widetilde{R}(\tau(\alpha), \alpha)(\widetilde{c}_{U}([p]), [p])
\end{equation}
is a symplectomorphism.   Moreover, the parametric  map
\begin{equation}\label{4.15}
\begin{aligned}
\widetilde{B}(\alpha) = & (id_{\mathbb{C}\mathbb{P}^{n-1}}, \widetilde{c}_{U})^{-1}\circ \widetilde{R}(\alpha)\\
= & (\widetilde{c}_{U}, id_{\mathbb{C}\mathbb{P}^{n-1}})^{-1}\circ (\widetilde{s}\circ \widetilde{R}_{\text{red}}(\tau(\alpha), \alpha))\circ(\widetilde{c}_{U}, id_{\mathbb{C}\mathbb{P}^{n-1}})\\
\end{aligned}
\end{equation}
is also a symplectomorphism, when $\mathbb{C}\mathbb{P}^{n-1}$ is equipped with $\omega_{\text{FS}}.$  Here 
$\widetilde{s}:\widetilde{\mathcal{G}}_{U}\longrightarrow \widetilde{\mathcal{G}}_{U}$
is the map induced by the permutation map on $ \mathbb{C}\mathbb{P}^{n-1}\times \mathbb{C}\mathbb{P}^{n-1}$ that maps $([p_1], [p_2])$ to 
$([p_2], [p_1]).$
\end{prop}
\begin{proof}
(a)  The assertion that $\widetilde{B}$ is an involution follows from \eqref{4.8} and Theorem 3.4 (a).   Moreover,
the relation in \eqref{4.11} follows from \eqref{3.25}.   Lastly, the explicit formula in \eqref{4.12} can be
obtained by direct computation by using $\widetilde{B}(\alpha) = j_{\delta}^{-1} \circ B^1(\alpha)\circ j_{\delta}$
and then by using the explicit formula for the simple elements.
\newline
(b)  We have the following relations
\begin{equation}\label{4.16}
\begin{aligned}
& \widetilde{R}(\alpha_1, \alpha_2) = (j_{\delta}\times j_{\delta})^{-1}\circ R^{1,1}(\alpha_1,\alpha_2)\circ (j_{\delta}\times j_{\delta}), \\
& \widetilde{B}_2(\alpha_2) =  (j_{\delta}\times j_{\delta})^{-1}\circ B_2^1(\alpha_2)\circ (j_{\delta}\times j_{\delta}),\\ 
& \widetilde{B}_1(\alpha_1) =  (j_{\delta}\times j_{\delta})^{-1}\circ B_1^1(\alpha_1)\circ (j_{\delta}\times j_{\delta}) \\
\end{aligned}
\end{equation}
from which we obtain
\begin{equation}\label{4.17}
\begin{aligned}
& \widetilde{B}_1(\alpha_1) \widetilde{R}_{21}(\tau(\alpha_2), \alpha_1)  \widetilde{B}_2(\alpha_2)\widetilde{R}_{12}(\alpha_1,\alpha_2) \\
= & (j_{\delta}\times j_{\delta})^{-1} \circ B_1^1(\alpha_1) R_{21}^{1,1}(\tau(\alpha_2), \alpha_1)  B_2^{1}(\alpha_2)R_{12}^{1,1}(\alpha_1,\alpha_2)\circ (j_{\delta} \times j_{\delta}).\\
\end{aligned}
\end{equation}
In a similar way, we find
\begin{equation}\label{4.18}
\begin{aligned}
& \widetilde{R}_{21}(\tau(\alpha_2), \tau(\alpha_1))\widetilde{B}_2(\alpha_2)\widetilde{R}_{12}(\tau(\alpha_1), \alpha_2) \widetilde{B}_1(\alpha_1)\\
= & (j_{\delta}\times j_{\delta})^{-1}\circ R_{21}^{1,1}(\tau(\alpha_2), \tau(\alpha_1))B_2^{1}(\alpha_2)R_{12}^{1,1}(\tau(\alpha_1), \alpha_2) B_1^{1}(\alpha_1)\circ (j_{\delta}\times j_{\delta}).\\
\end{aligned}
\end{equation}
The assertion therefore follows from \eqref{3.26} and the above relations.
\newline
(c)   In view of the relation
\begin{equation}\label{4.19}
\widetilde{R}(\alpha) = \widetilde{R}_{\text{red}}(\tau(\alpha), \alpha)\circ (\widetilde{c}_{U}, id_{\mathbb{C}\mathbb{P}^{n-1}})
\end{equation}
and Proposition 4.1, it suffices to show that $(\widetilde{c}_{U}, id_{\mathbb{C}\mathbb{P}^{n-1}})$
is a symplectomorphism,  when its domain is equipped with $2\omega_{\text{FS}},$ and its codomain
is equipped with $\widetilde{\omega}.$  This checking is as in \eqref{3.49}.   Finally, it is straightforward
to derive the first line of \eqref{4.15}.   From this relation and what we just proved, it is now plain
that $\widetilde{B}(\alpha)$ is a symplectomorphism, when $\mathbb{C}\mathbb{P}^{n-1}$ is
equipped with $\omega_{FS}.$
\end{proof}
In view of Proposition 4.3, $\widetilde{B}(\alpha)$ is a parametric reflection map.

Now we fix a positive integer $N.$  For given $\alpha_1,\cdots, \alpha_N\in \mathbb{C}\setminus (\mathbb{R}\cup \sqrt{-1}\mathbb{R}),$
 let $\alpha_{i+N} = \tau(\alpha_i),$
$i=1,\cdots, N.$   We will make the assumption that $\{\alpha_i, \overline{\alpha}_i\} \cap \{\alpha_j, \overline{\alpha}_j\} =\emptyset,$
for $i\neq j,$ $1\leq i,j\leq 2N.$  Given $[p_1^{-}],\cdots, [p_{N}^{-}]$ in $\mathbb{C}\mathbb{P}^{n-1},$ consider
the refactorization problem
\begin{equation}\label{4.20}
\begin{aligned}
& \sigma(g_{\alpha_1, \pi_{[p_1^{-}]}}\cdots g_{\alpha_N, \pi_{[p_N^{-}]}}) g_{\alpha_1, \pi_{[p_1^{-}]}}\cdots g_{\alpha_N, \pi_{[p_N^{-}]}} \\
= & g_{\alpha_N, \pi_{[p_N^{+}]}}\cdots g_{\alpha_1, \pi_{[p_1^{+}]}} \sigma(g_{\alpha_N, \pi_{[p_N^{+}]}}\cdots g_{\alpha_1, \pi_{[p_1^{+}]}}).
\\
\end{aligned}
\end{equation}
Under the above assumptions, the problem has unique solutions for $[p_1^+],\cdots , [p_N^+].$  We define
\begin{equation}\label{4.21}
\begin{aligned}
 &\Pi(\alpha_1,\cdots, \alpha_N):   (\mathbb{C}\mathbb{P}^{n-1})^N\longrightarrow (\mathbb{C}\mathbb{P}^{n-1})^N\\
 \,\,\,& ([p_1^{-}],\cdots, [p_N^{-}])\mapsto ([\widetilde{c}_{U}([p_1^+]),\cdots , \widetilde{c}_{U}([p_N^+])) .\\
 \end{aligned}
 \end{equation}  
 We will call $\Pi(\alpha_1,\cdots, \alpha_N)$ the $N$-body polarization reflection map.  Note that when
 $N=1,$  $\Pi(\alpha_1)$ is nothing but the parametric reflection map $\widetilde{B}(\alpha_1).$  When
 $N=2,$  it follows from the calculations in \eqref{3.40} and \eqref{3.45} that $\Pi(\alpha_1,\alpha_2)$  can be obtained as the composite  
 \begin{equation}\label{4.22}
 \widetilde{B}_1(\alpha_1) \widetilde{R}_{21}(\tau(\alpha_2), \alpha_1)\widetilde{B}_2(\alpha_2)\widetilde{R}_{12}(\alpha_1,\alpha_2).
 \end{equation}
 Thus $\Pi(\alpha_1,\alpha_2)$ is just  the map
 which appears in the formulation of the parametric reflection equation \eqref{4.13}.
\begin{theorem}\label{T: 4.4}
Under the above assumptions, the $N$-body polarization reflection map $\Pi(\alpha_1,\cdots, \alpha_N)$ is a symplectomorphism, when $(\mathbb{C}\mathbb{P}^{n-1})^N$ is equipped
with the symplectic $2$-form
\begin{equation}\label{4.23}
\Omega_{\alpha_1,\cdots, \alpha_N}=(\overline{\alpha}_1-\alpha_1)\omega_{\text{FS}}\oplus \cdots \oplus (\overline{\alpha}_N -\alpha_N)\omega_{\text{FS}}.
\end{equation}
\begin{proof}
Let $\widetilde{B}_i(\alpha_i),$ $1\leq N$ be the map from $(\mathbb{C}\mathbb{P}^{n-1})^{N}$ to itself which acts as $\widetilde{B}(\alpha_i)$ on the
$i$-th factor of $(\mathbb{C}\mathbb{P}^{n-1})^N$ and acts as the identity on the other factors.  Similarly, we define $\widetilde{R}_{ij}(\alpha_i,\alpha_j)$  (resp.~ $\widetilde{R}_{ij}(\tau(\alpha_i), \alpha_j)$)
for $1\leq i\leq j\leq N$ (resp.~ for $\leq 1\leq j<i \leq N$ ) as the map from $(\mathbb{C}\mathbb{P}^{n-1})^{N}$ to itself which acts as $\widetilde{R}(\alpha_i,\alpha_j)$
on the $i$-th and the $j$-th factor of $(\mathbb{C}\mathbb{P}^{n-1})^{N}$ and as the identity on the other factors.  Since the unique solution of the refactorization
problem can be obtained by applying Theorem 2.2 (a) repeatedly, it follows that we can obtain the map $\Pi(\alpha_1,\cdots, \alpha_N)$
as a composition of maps of the three types which we introduced above.   Now, from Theorem 2.3 (b), for each $(i,j)$ with $1\leq i<j\leq N$ (resp.~$1\leq j<i\leq N$)  the map $\widetilde{R}_{ij}(\alpha_i, \alpha_j)$ (resp.~ $\widetilde{R}_{ij}(\tau(\alpha_i), \alpha_j)$) is a symplectomorphism, when $(\mathbb{C}\mathbb{P}^{n-1})^{N}$ is equipped with the structure in \eqref{4.23}.   On the other hand, it follows from Proposition 4.3 (c) that for each $1\leq i\leq N,$ $\widetilde{B}_i(\alpha_i)$ is also a symplectomorphism of $((\mathbb{C}\mathbb{P}^{n-1})^{N}, \Omega_{\alpha_1,\cdots, \alpha_N}).$  Hence the assertion follows.
\end{proof}
\end{theorem}
 We will wrap up this section by explaining the physical meaning of the polarization reflection map in the context of the $n$-Manakov system on the half-line $x\geq 0$ with mixed Dirichlet/Neumann boundary conditions at $x=0.$   We will also point out, for general Hermitian $U\in U(n),$ the relationship between the {\sl full} polarization
 scattering map in \eqref{2.30}, and that of the $N$-body polarization reflection map in \eqref{4.21}.
 
 Consider an $N$-soliton solution of the half-line problem, with mixed Dirichlet/Neumann boundary conditions at $x=0.$   As we pointed out earlier, this can be obtained from a $2N$-soliton solution on the full line with norming constants and the poles $\alpha_j$ satisfying certain mirror symmetry conditions \cite{CZ2}.    More precisely, when $t\to -\infty,$ the $2N$-soliton behaves like
the sum of $2N$ $1$-$soliton$ solutions characterized by $\alpha_j = \frac{1}{2}(u_j + iv_j)$ for the real solitons and $\tau(\alpha_j)=-\overline{\alpha}_j$ for the
`mirror' solitons, where $u_j, v_j >0$ for $j=1,\cdots, N.$   Thus $-2u_j$
is the velocity of the real $j$-th $1$-soliton on $x>0$ while  $2u_j$ is the velocity of the $j$-th `mirror' $1$-soliton on $x<0,$ $j=1,\cdots, N.$  Assume 
\begin{equation}\label{4.24}
0< u_1< u_2<\cdots < u_N
\end{equation}
so that the $1$-$solitons$ are arranged in the order $1,\cdots, N$  on the positive $x$-axis as $t\to -\infty.$   The system evolves towards the boundary at $x=0,$
where the real solitons interact with the `mirror' solitons which then turn into real solitons.    To summarize, we have the following scattering picture \cite{CZ1}:
\begin{equation}\label{4.25}
\begin{aligned}
& 2N, 2N-1,\cdots, N+1 | 1, 2,\cdots , N,\quad t\to -\infty,\\
& N, N-1,\cdots, 1 | N+1, N+2,\cdots, 2N, \quad t\to \infty,\\
\end{aligned}
\end{equation}
where the vertical bar stands for the boundary at $x=0.$
Consequently, the polarization scattering map is given by
\begin{equation}\label{4.26}
([U p_N^{-}],\cdots, [U p_1^{-}], [p_1^{-}],\cdots, [p_{N}^{-}])\mapsto ([U p_N^{+}],\cdots, [ U p_1^{+}], [p_1^{+}],\cdots, [p_N^{+}])
\end{equation}
with $U=VI_{S}V^*$,  where $p_j^{-}$ (resp.~ $U p_j^{+}$) is the asymptotic unit polarization vector of the $j$-th real soliton as $t\to-\infty$ (resp.~ $t\to \infty$), $j=1,\cdots, N.$

With this information, we can now interpret the corresponding $N$-body polarization reflection map, given by $\Pi(\alpha_1,\cdots, \alpha_N),$ as defined in \eqref{4.21}, in the case where $U=V I_{S}V^*$ and with the above
assumptions.   And the general definition in \eqref{4.21} is motivated by this scenario.

To understand the relationship between the polarization scattering map in \eqref{4.26} and the $N$-body polarization reflection map, we introduce further notations.
For this purpose, let $(\widetilde{c}_{U})^N = \widetilde{c}_{U}\times \cdots \widetilde{c}_{U}$  ($N$ copies), and let 
\begin{equation}\label{4.27}
\beta: (\mathbb{C}\mathbb{P}^{n-1})^N\longrightarrow (\mathbb{C}\mathbb{P}^{n-1})^N,\,\,([p_1],\cdots, [p_N])\mapsto ([p_N],\cdots, [p_1]).
\end{equation}
Denote the graph of $(\widetilde{c}_{U})^N\circ \beta$ by
\begin{equation}\label{4.28}
\widetilde{\mathcal{G}}_{U}^N=\{ ((\widetilde{c}_{U})^N\circ \beta(x), x)\mid x\in (\mathbb{C}\mathbb{P}^{n-1})^N\}.
\end{equation}
Then the polarization scattering map in \eqref{4.26} is the reduction 
\begin{equation}\label{4.29}
\mathcal{S}_{\text{red}}(\tau(\alpha_N),\cdots, \tau(\alpha_1),\alpha_1,\cdots, \alpha_N): \widetilde{\mathcal{G}}_{U}^N\longrightarrow
\widetilde{\mathcal{G}}_{U}^N .
\end{equation}
Therefore, if $\widetilde{s}_N: \widetilde{\mathcal{G}}_{U}^N\longrightarrow  \widetilde{\mathcal{G}}_{U}^N$ is the map induced by the permutation map
on $(\mathbb{C}\mathbb{P}^{n-1})^N\times (\mathbb{C}\mathbb{P}^{n-1})^N$ that maps $(X,Y)$ to $(Y,X),$ we have
\begin{equation}\label{4.30}
\begin{aligned}
& \Pi(\alpha_1,\cdots, \alpha_N)\\
= & \beta\circ ((\widetilde{c}_{U})^N\circ \beta, id_{(\mathbb{C}\mathbb{P}^{n-1})^N})^{-1} (\widetilde{s}_N\circ \mathcal{S}_{\text{red}}(\tau(\alpha_N),\cdots, \tau(\alpha_1),\alpha_1,\cdots, \alpha_N))\circ ((\widetilde{c}_{U})^N\circ \beta, id_{(\mathbb{C}\mathbb{P}^{n-1})^N}),\\
\end{aligned}
\end{equation}
where $\alpha$ stands for $(\alpha_1,\cdots, \alpha_N)$ and similarly $\tau(\alpha)$ is the shorthand for $(\tau(\alpha_1),\cdots, \tau(\alpha_N)).$   Note that in the case when $N=1,$ the map $\beta$ is just the
identity map and so in this case, the relation in \eqref{4.30} is just the relation in \eqref{4.15}.   It is in this
sense that we have a generalization of the relation in \eqref{4.15}. 
\begin{remark}\label{R: 4.5}
For given projectors $P_1^{-}\in P(n)_{k_1},\cdots, P_N^{-}\in P(n)_{k_N},$ we can consider the refactorization problem
\begin{equation}\label{4.31}
\begin{aligned}
& \sigma(g_{\alpha_1, P_1^{-}}\cdots g_{\alpha_N, P_N^{-}}) g_{\alpha_1,  P_1^{-}}\cdots g_{\alpha_N, P_N^{-}} \\
= & g_{\alpha_N, P_N^{+}}\cdots g_{\alpha_1, P_1^{+}} \sigma(g_{\alpha_N, P_N^{+}}\cdots g_{\alpha_1, P_1^{+}}).\\
\end{aligned}
\end{equation}
Clearly, we can formulate the analog of Theorem 4.4 in this context, by introducing the map
\begin{equation}\label{4.32}
(P_1^{-},\cdots P_N^{-})\mapsto (c_{U,k_1}(P_1^{+}),\cdots , c_{U, k_N}(P_N^{+})).
\end{equation}
We will leave the details to the interested reader.
\end{remark}

%%%%%%%%%%%%%%%%%%%%%%%%%%%%%%%%%%%%%%%%%%%%%%%%%%%%%%%%%%%%%%%%%%%%%%%%%%%%%%%%%%%%%
\section{Reflection maps and Poisson Lie groups}
%%%%%%%%%%%%%%%%%%%%%%%%%%%%%%%%%%%%%%%%%%%%%%%%%%%%%%%%%%%%%%%%%%%%%%%%%%%%%%%%%%%%

We begin by formulating two results in the context of a Poisson Lie group $(G, \pi_{G}).$  Then we will apply the general result to the rational 
loop group $K_{\text{rat}}.$

We will make the following assumptions.
\vskip .1in
\noindent $\underline{(A1)}$  There exist a left partial group action $\xi$ of $G$ on itself given by a family of subsets $\{G_g\}_{g\in G}$ of $G$ and a family of bijections
$\{\xi_{g}: G_{g^{-1}}\longrightarrow G_g\}_{g\in G}$ satisfying the usual conditions \cite{B, L1}.  We also assume the existence of  a right partial group action $\eta$ of 
$G$ on itself given by the same family of subsets above and a family of bijections $\{\eta_{g}: G_{g^{-1}}\longrightarrow G_g\}_{g\in G}.$
\newline
\noindent $\underline{(A2)}$  Let 
\begin{equation}\label{5.1}
G\ast G =\{ (g,h)\in G\times G\mid g\in G_{h^{-1}}, h\in G_{g^{-1}}\}.
\end{equation}
We assume $G\ast G$ is an open submanifold of $G\times G$ and that the partial group actions are compatible in the sense that
\begin{equation}\label{5.2}
gh = \xi_{g} (h) \eta_{h} (g) \quad \text{for all} \,\,\, (g,h)\in G\ast G .
\end{equation}
In addition, we assume that
\begin{equation}\label{5.3}
G_{g^{-1}} = G_{g}\quad \text{for all}\,\, g\in G
\end{equation}
so that  $\xi_{g}$ and $\eta_{g}$ are maps from $G_{g}$ to itself.   

From the definition of $G\ast G,$ it is clear that
 $G\ast G$ is symmetric in the sense that
\begin{equation}\label{5.4}
(g,h)\in G\ast G \iff (h, g) \in G\ast G.
\end{equation}
\newline
\noindent $\underline{(A3)}$  For all $(g, h) \in G\ast G,$ we assume that
\begin{equation}\label{5.5}
G_{\xi_{g}(h)} = G_{h}, \quad G_{\eta_{h}(g)} = G_{g}.
\end{equation}
In addition,  we assume that
\begin{equation}\label{5.6}
g_1\in G_{g_2} \iff g_2\in G_{g_1}\quad \text{for}\,\, g_1, g_2\in G.
\end{equation}
\noindent  $\underline{(A4)}$  Let $\sigma: G\longrightarrow G$ be a Poisson involution which is also a Lie group anti-morphism such that
\begin{equation}\label{5.7}
(g, h) \in G\ast G \implies (\sigma(h), \sigma(g)) \in G\ast G, 
\end{equation}
and moreover, 
\begin{equation}\label{5.8}
\sigma(\xi_{g}(h)) = \eta_{\sigma(g)}(\sigma(h)),  \,\, \sigma(\eta_{h}(g)) = \xi_{\sigma(h)}(\sigma(g)).
\end{equation}
Let the graph of
$\sigma$ be the submanifold
\begin{equation}\label{5.9}
G(\sigma):=\{(\sigma(g), g)\mid g\in G\}.
\end{equation}
We assume
\begin{equation}\label{5.10}
G^{\prime}:=(\sigma, id_{G})^{-1}(G(\sigma)\cap (G\ast G))\neq \emptyset,
\end{equation}
where $(\sigma, id_{G})$ is the diffeomorphism defined by
\begin{equation}\label{5.11}
 (\sigma, id_{G}) : G\longrightarrow G(\sigma), g\mapsto (\sigma(g), g).
\end{equation}
In addition, we assume that 
\begin{equation}\label{5.12}
\text{the equation}\,\, x \sigma(x) =1\,\,\, \text{has only the solution}\,\,\, x=1.
\end{equation}
\begin{remark}
If in assumption (A1), $\xi$ and $\eta$ are group actions (rather than partial group actions), then we have $G_{g} = G$ for all $g\in G.$   In this case, we
can replace $G\ast G$ in assumption (A2) by the product $G\times G$ and (5.3) will hold automatically.   Also, there is no need to make the
assumptions in (A3).    However, in (A4), we still have to assume the relations in (5.8) and that
the equation $x\sigma(x)=1$ has only the trivial solution $x=1.$    The reader will see that they are essential in our discussion below.
\end{remark}

In the following, we will equip $G\times G$ with the product Poisson structure so that the open submanifold $G\ast G$ of $G\times G$ has an induced Poisson
bracket.  Similarly, we will equip the open submanifold $G^{\prime}$ of $G$ with the induced Poisson bracket.   
\begin{definition}\label{D:5.2}
We define
\begin{equation}\label{5.13}
R: G\ast G\longrightarrow G\ast G
\end{equation}
by the formula
\begin{equation}\label{5.14}
R(g, h) = (\eta_{h}(g), \xi_{g}(h)),\,\,\, (g,h)\in G\ast G
\end{equation}
and define the (putative) {\sl reflection map} $\mathbf{B}$ associated with $R$ and $\sigma$ by
\begin{equation}\label{5.15}
\mathbf{B}: G^{\prime}\longrightarrow G^{\prime}, g\mapsto \eta_{g} (\sigma(g)).
\end{equation}
\end{definition}
\begin{remark}\label{R:5.3}
(a) In the definition of $\mathbf{B},$ the reason why we know $\eta_{g} (\sigma(g))\in G^{\prime}$ follows from $\sigma(\eta_{g} (\sigma(g))=\xi_{\sigma(g)}(g)\in G_{\sigma(g)} = G_{\eta_{g}(\sigma(g))}$ and $\eta_{g}(\sigma(g))\in G_{g} = G_{\xi_{\sigma(g)}(g)},$ where we have used the definition of $G^{\prime}$ and (5.5).
\newline
(b) Defining the reflection map $\mathbf{B}$ by the formula in \eqref{5.15} is a matter of choice we pick here.   The fact is that instead of $\mathbf{B},$ we could use $\mathbf{B}^{\prime}: G^{\prime}\longrightarrow G^{\prime}, \sigma(g)\mapsto \xi_{\sigma(g)}(g)$  because of the mirror symmetry.    Indeed, if we call $\mathbf{B}$ the reflection map, then 
 we might call $\mathbf{B}^{\prime}$ the {\sl mirror reflection map}, as motivated by the collisions scenario in Section 4  (since we can consider the map which keeps
 track of the change in polarization vector of the mirror $1$-soliton).
\end{remark}
\begin{theorem}\label{T: 5.3}   Under assumptions (A1)-(A3) above,  
\newline
(a) $R$ is a Yang-Baxter map, i.e.,
\begin{equation}\label{5.16}
R_{12} R_{13} R_{23} = R_{23} R_{13} R_{12},
\end{equation}
where we interpret \eqref{5.16} as an equality of maps from $G^{(3)}$ to itself, where
\begin{equation}\label{5.17}
G^{(3)} =\{ (g_1, g_2, g_3)\in G\times G\times G\mid (g_i, g_j) \in G\ast G, \,i\neq j\}.
\end{equation}
(b)  $R$ is a Poisson diffeomorphism,  when the open Poisson submanifold $G\ast G$ is equipped with
the structure induced from $G\times G.$
\end{theorem}
\begin{proof}   (a)  Take $(g_1, g_2, g_3) \in G^{(3)}.$  Then clearly, $R_{23} (g_1, g_2, g_3)$ is defined.      In order for $R_{13} R_{23}(g_1, g_2,g_3)$ to be
defined, we require that $g_1\in G_{\xi_{g_2}(g_3)}$ and $\xi_{g_2}(g_3)\in G_{g_1}.$   As $G_{\xi_{g_2}(g_3)}= G_{g_3}$ by (A3), the first condition is 
satisfied and so the second condition follows by \eqref{5.6} in (A3).   This gives
\begin{equation}\label{5.18}
R_{13} R_{23}(g_1, g_2, g_3) = (\eta_{\xi_{g_2}(g_3)}(g_1), \eta_{g_3}(g_2), \xi_{g_1 g_2}(g_3)).
\end{equation}
To show that $R_{12}R_{13} R_{23}(g_1, g_2, g_3)$ is defined, we require that $\eta_{\xi_{g_2}(g_3)}(g_1)\in G_{\eta_{g_3}(g_2)}= G_{g_2}$ and 
$\eta_{g_3}(g_2)\in G_{\eta_{\xi_{g_2}(g_3)}(g_1)}= G_{g_1}.$   However, the former follows as we have $g_2 \in G_{g_1} =  G_{\eta_{\xi_{g_2}(g_3)}(g_1)}$
and so the second one is true as well by (5.6) in (A3).     So this leads to
\begin{equation}\label{5.19}
R_{12} R_{13} R_{23}(g_1, g_2, g_3) = (\eta_{g_2 g_3}(g_1), \xi_{\eta_{\xi_{g_2}(g_3)}(g_1)}(\eta_{g_3}(g_2)), \xi_{g_1g_2}(g_3)).
\end{equation}
In a similar way, we can show that $R_{23} R_{13}R_{12}(g_1, g_2, g_3)$ is defined and we have
\begin{equation}\label{5.20}
R_{23} R_{13} R_{12} (g_1,g_2,g_3) = (\eta_{g_2g_3}(g_1), \eta_{\xi_{\eta_{g_2}(g_1)}(g_3)}(\xi_{g_1}(g_2)), \xi_{g_1g_2}(g_3))
\end{equation}
The argument to show that the expressions in \eqref{5.19} and \eqref{5.20} are equal is identical to the one in Corollary 5.2 of \cite{L1}.
\newline
(b)  The proof follows the same argument as in the proof of Theorem 5.13 in \cite{L1}.
\end{proof}

\begin{lemma}\label{L: 5.5}
 The map 
\begin{equation}\label{5.21}
\Sigma: G\times G\longrightarrow G\times G,  (g,h) \mapsto (\sigma(h), \sigma(g)), (g,h)\in G\times G
\end{equation}
is a Poisson involution with stable locus $(G\times G)^{\Sigma}$ given by the graph of $\sigma,$ defined in
\eqref{5.9} above.
Hence $G(\sigma)$ is a Dirac submanifold of $G\times G$ and the bundle map of its induced Poisson structure is given by
the formula
\begin{equation}\label{5.22}
\begin{aligned}
& \pi^{\#}_{G(\sigma)} (\sigma(g), g) (a, b)\\
 =  & \frac{1}{2} ( \pi^{\#}_{G}(\sigma(g)) (a) +  \pi^{\#}_{G} (\sigma(g)) T_{\sigma(g)}^{*} \sigma (b), 
  \pi^{\#}_{G}(g)(b) + \pi^{\#}_{G}(g) T^{*}_{g} \sigma (a)). \\
\end{aligned}
\end{equation}
Consequently, the open submanifold $G^{\prime}(\sigma)$ of $G(\sigma)$ defined by
\begin{equation}\label{5.23}
G^{\prime}(\sigma) := G(\sigma)\cap (G \ast G)
\end{equation}
carries an induced Poisson structure $\pi_{G^{\prime}(\sigma)}.$
\end{lemma}
\begin{proof}
 Let $s: G\times G\longrightarrow G\times G$ be the swap map, given by $s(g, h) = (h, g),$ $(g, h)\in G\times G.$   Then
clearly $\Sigma = (\sigma\times \sigma)\circ s.$   Since both $s$ and $\sigma\times \sigma$ are Poisson involutions, it
follows that $\Sigma$ is a Poisson involution and the assertion about its stable locus is clear.   To compute the bundle map of the induced Poisson structure
on the Dirac submanifold $(G\times G)^{\Sigma} = G(\sigma),$ we make use of the formula
\begin{equation}\label{5.24}
\pi^{\#}_{G(\sigma)} = pr \circ \pi^{\#}_{G\times G}\mid_{G(\sigma)} \circ\, pr^{*},
\end{equation}
where $pr: T_{G(\sigma)}(G\times G)\longrightarrow T G(\sigma)$ is the projection map induced by the vector bundle decomposition
\begin{equation}\label{5.25}
T_{G(\sigma)} (G\times G) = T G(\sigma) \oplus \bigcup _{(\sigma(g), g)\in G(\sigma)} \text{ker}\,(T_{(\sigma(g), g)} \Sigma +1).
\end{equation}
From
\begin{equation}\label{5.26}
\begin{aligned}
& T_{(\sigma(g), g)} G(\sigma) =\{ (T_g \sigma (v), v) \mid v\in T_{g} G\},\\
& \text{ker}\,(T_{(\sigma(g), g)} \Sigma +1)=\{ (-T_{g} \sigma (v), v) \mid v\in T_{g} G\},\\
\end{aligned}
\end{equation} 
a direct calculation shows that
\begin{equation}\label{5.27}
pr_{(\sigma(g), g)} (v, w) = \frac{1}{2} ( v + T_{g}\sigma (w),  w + T_{\sigma(g)} \sigma (v)).
\end{equation}
Using this, another computation gives 
\begin{equation}\label{5.28}
pr^{*}_{(\sigma(g), g)} (a,b) = \frac{1}{2} (a + T_{\sigma(g)}^{*} \sigma (b),  b + T_{g}^{*}\sigma (a)).
\end{equation}
By making use of  \eqref{5.24}, \eqref{5.27}, and \eqref{5.28}, a straightforward but lengthy
calculation then gives the formula in \eqref{5.22}.   Lastly, it follows from (A4) that $G^{\prime}(\sigma)\neq \emptyset,$ hence the assertion is clear.
\end{proof}
In order to analyze $\bold{B},$  we introduce the diffeomorphism
\begin{equation}\label{5.29}
 (\sigma, id_{G})^{\prime} : G^{\prime} \longrightarrow G^{\prime}(\sigma), g\mapsto (\sigma(g), g).
\end{equation}
If $i_{G^{\prime}(\sigma)}: G^{\prime}(\sigma)\longrightarrow G(\sigma)$ and $i_{G^{\prime}}: G^{\prime} \longrightarrow G$ are the embedding maps, we 
have the relation
\begin{equation}\label{5.30}
i_{G^{\prime}(\sigma)}\circ (\sigma, id_{G})^{\prime} = (\sigma, id_{G})\circ i_{G^{\prime}}.
\end{equation}

Clearly, we can push the Poisson structure on $G$ forward to $G(\sigma)$ using this map so that $(\sigma, id_{G})$ is a Poisson diffeomorphism
when its codomain is equipped with the pushforward structure.    We now compute this structure and describe its consequences.

\begin{lemma}\label{L: 5.6}
(a)  For all $g\in G,$ $(a,b)\in T_g^{*} G(\sigma),$ we have
\begin{equation}\label{5.31}
\begin{aligned}
& T_g (\sigma, id_{G})\circ \pi^{\#}_{G} (g) \circ T_{g}^{*}(\sigma, id_{G}) (a,b)\\
= & ( T_g\sigma\pi^{\#}_{G}(g)T^*_g\sigma (a) + T_g\sigma \pi^{\#}_{G}(g)(b), 
 \pi^{\#}_{G}(g)T^*_{g}\sigma (a) +  \pi^{\#}_{G}(g) (b))\\
= & 2 \pi^{\#}_{G(\sigma)}(\sigma(g), g)(a,b)
\end{aligned}
\end{equation}
so that $(\sigma, id_{G}): (G, \pi_{G})\longrightarrow (G(\sigma), 2 \pi_{G(\sigma)})$ is a Poisson diffeomorphism.
\newline
(b)  The map $(\sigma, id_{G})^{\prime} : (G^{\prime}, \pi_{G^{\prime}})\longrightarrow (G^{\prime}(\sigma), 2 \pi_{G^{\prime}(\sigma)})$
is a Poisson diffeomorphism.
\end{lemma}  
\begin{proof}
(a)  We have the formulas
\begin{equation}\label{5.32}
T_g (\sigma, id_{g}) (v) = (T_{g}\sigma(v), v), \quad T_{g}^{*} (\sigma, id_{G})(a,b) = b + T_{g}^{*}\sigma(a)
\end{equation}
from which we obtain the second line in \eqref{5.31}.   To pass from the second line to the last line in \eqref{5.31}, we use
the fact that $\sigma$ is a Poisson map from which we find $T_{g}\sigma \pi^{\#}_{G}(g) = \pi^{\#}_{G}(\sigma(g))T^{*}_{\sigma(g)}\sigma.$
The assertion therefore follows by comparing with the formula in \eqref{5.22}.
\newline
(b)  This is a consequence of part (a) and the relation in \eqref{5.30}.
\end{proof}
With this preparation, we are now ready to establish the following.   For this purpose, consider the reduction of the
map $\Sigma$ in \eqref{5.21} to $G\ast G$:
\begin{equation}\label{5.33}
\Sigma\mid_{G\ast G} : G\ast G\longrightarrow G\ast G,
\end{equation}
in the sense of Theorem 2.14, which is well defined by (A4).   Moreover, it is a Poisson involution since $\Sigma$ is a Poisson involution and the embedding map of $G\ast G$ into $G\times G$
is Poisson.
\begin{theorem}\label{T: 5.7}  (a)  The map $R$ commutes with $\Sigma\mid{G\ast G},$ i.e. $\Sigma\mid_{G\ast G}\circ R = R\circ \Sigma\mid_{G\ast G}$ and therefore its
reduction
\begin{equation}\label{5.34}
R_{\text{red}} : G^{\prime}(\sigma)\longrightarrow G^{\prime}(\sigma)
\end{equation}
is a Poisson diffeomorphism, when $G^{\prime}(\sigma)$ is equipped with any nonzero multiple of $\pi_{G^{\prime}(\sigma)}.$  That is,
$R_{\text{red}}$ is a Dirac reduction of $R.$
\newline
(b)  Let $s: G\times G\longrightarrow G\times G$ be the swap map, given by $s(g,h) = (h,g)$ for $(g,h)\in G\times G.$ Then the map
\begin{equation}\label{5.35} 
\bold{B} : (G^{\prime}, \pi_{G^{\prime}}) \longrightarrow (G^{\prime}, \pi_{G^{\prime}})
\end{equation}
satisfies the relation
\begin{equation}\label{5.36}
\bold{B} = ((\sigma, id_{G})^{\prime})^{-1}\circ (s\circ R)_{\text{red}} \circ (\sigma, id_{G})^{\prime}
\end{equation}
and hence is a Poisson diffeomorphism.  Here $(s\circ R)_{\text{red}} : G^{\prime}(\sigma)\longrightarrow G^{\prime}(\sigma).$
\newline
(c)  The map $\bold{B}$ satisfies the reflection equation
\begin{equation}\label{5.37}
\bold{B}_1 R_{21} \bold{B}_2 R_{12} = R_{21}\bold{B}_2 R_{12} \bold{B}_1
\end{equation}
where we interpret \eqref{5.37} as an equality of maps from $G^{(2)}_{\sigma}$ to itself, where
\begin{equation}\label{5.38}
G^{(2)}_{\sigma}=\{(g_1, g_2)\in G^{\prime}\times G^{\prime}\mid (g_1, g_2)\in G\ast G, (g_1, \sigma(g_2))\in G\ast G\}.
\end{equation}
Hence $\bold{B}$ is a reflection map.
\end{theorem}
\begin{proof}
(a) Let $(g, h)\in G\ast G.$   Then from $gh =\xi_{g}(h) \eta_{h}(g),$ we have $\Sigma\mid_{G\ast G}\circ R(g,h) = (\sigma(\xi_{g}(h)), \sigma(\eta_{h}(g)).$   On the other
hand, since $\sigma$ is a Lie group anti-morphism, it follows that $\sigma(h)\sigma(g) =\sigma(\eta_{h}(g)) \sigma(\xi_{g}(h)).$  From this, we find that
$R\circ \Sigma\mid_{G\ast G}(g,h)= R(\sigma(h), \sigma(g))= (\sigma(\xi_{g}(h)), \sigma(\eta_{h}(g)).$  Since $(g, h)\in G\ast G$ is arbitrary, we thus conclude that
$R$ commutes with the Poisson involution $\Sigma\mid_{G\ast G}.$   Since the stable locus of $\Sigma\mid_{G\ast G}$ is given by
$G^{\prime}(\sigma),$ it follows from Theorem 5.4 (b) above and  Dirac reduction (Corollary 2.9) 
that the map $R\mid_{G^{\prime}(\sigma)}$ is a Poisson
diffeomorphism, when $G^{\prime}(\sigma)$ is equipped with the induced structure in (5.22).
\newline
(b)  In view of the last relation in \eqref{5.8},  the map that sends $R(\sigma(g), g)$ to $\eta_{g}(\sigma(g))$ is given by $(id_{G},\sigma)^{-1}.$  But clearly,
$(id_{G}, \sigma) = s\mid_{G(\sigma)}\circ (\sigma, id_{G}).$  As $s$ is an involution,  the relation in \eqref{5.36} follows.   Now by Lemma 5.6, the
map $(id_{G}, \sigma): (G, \pi_{G})\longrightarrow (G(\sigma), 2\pi_{G(\sigma)})$ is a Poisson map.   On the other hand, it follows from part (a)
above that $R\mid_{G(\sigma)}$ is a Poisson map, when $G(\sigma)$ is equipped with the structure $2 \pi_{G(\sigma)}.$   As 
$s\circ \Sigma = \Sigma\circ s = (\sigma\times \sigma)\mid_{G\ast G},$ as can be easily verified, it follows by Dirac reduction that
$s\mid_{G(\sigma)}$ is Poisson, when $G(\sigma)$ is equipped with the structure $2 \pi_{G(\sigma)}.$   Lastly, it follows from the 
above discussion that $(id_{G}, \sigma)^{-1}$ is a Poisson map from $(G(\sigma), 2\pi_{G(\sigma)})$  to $(G, \pi_{G}).$
Since composition of Poisson maps is Poisson, the assertion regarding $B^{\sigma}$ follows from \eqref{5.36}.
\newline
(c)  We will establish the relation
\begin{equation}\label{5.39}
\begin{aligned}
& R_{31} R_{32} R_{41} R_{42} R_{43} R_{12} (g_1, g_2, \sigma(g_1), \sigma(g_2)) \\
= & R_{43} R_{12} R_{42} R_{32} R_{41} R_{31}(g_1, g_2, \sigma(g_1), \sigma(g_2))\\
\end{aligned}
\end{equation}
under the assumption that $(g_1, g_2)\in G^{(2)}_{\sigma}.$   We begin by showing that the first line above is well defined and in the process,
we will compute the expression step by step.  Since $(g_1, g_2)\in G\ast G,$ and (5.8) holds, we have
the factorizations
\begin{equation}\label{5.40}
\begin{aligned}
& g_{1} g_{2} = h_{2} h_{1}, \quad h_{2} = \xi_{g_1}(g_2), h_1 = \eta_{g_2}(g_1)\\
& \sigma(g_{2}) \sigma(g_1) = \sigma(h_{1}) \sigma(h_{2}), \quad \sigma(h_1) = \xi_{\sigma(g_2)}(\sigma(g_1)), \sigma(h_2) = \eta_{\sigma(g_1)}(\sigma(g_2))\\
\end{aligned}
\end{equation}
from which it follows that
\begin{equation}\label{5.41}
R_{43} R_{12} (g_1, g_2, \sigma(g_1), \sigma(g_2))= (h_1, h_2, \sigma(h_1), \sigma(h_2)).
\end{equation}
On the other hand, since $g_2\in G^{\prime},$ we can check that $h_2\in G^{\prime},$ and hence we have the factorization
\begin{equation}\label{5.42}
\sigma(h_2) h_2 = j_2 \sigma(j_2),\quad j_2 = \xi_{\sigma(h_2)} (h_2)
\end{equation}
so that
\begin{equation}\label{5.43}
R_{42} R_{43} R_{12} (g_1, g_2, \sigma(g_1), \sigma(g_2))= (h_1, j_2, \sigma(h_1), \sigma(j_2)).
\end{equation}
Now we want to apply $R_{41}$ and $R_{32}$ to the expression above.    In order to be able to do this, we form $\sigma(j_2) h_1,$  and for
solvability of the refactorization problem, we require that
$\sigma(j_2)\in G_{h_1}$ and $h_1\in G_{\sigma(j_2)}.$  By symmetry, and by using (5.5), (5.36), and (5.34), it suffices to show that 
$h_1 =\eta_{g_2}(g_1)\in G_{\sigma(j_2)} = G_{\sigma(h_2)}= G_{\sigma(g_2)}.$  Again by symmetry, it suffices to show that
$\sigma(g_2) \in G_{\eta_{g_2}(g_1)} = G_{g_1}.$   But the validity of this follows by assumption that $(g_1, \sigma(g_2))\in G\ast G.$
Thus we have
\begin{equation}\label{5.44}
\sigma(j_2) h_1 = k_1\sigma(\ell_2), \,\, \sigma(h_1) j_2 = \ell_2 \sigma(k_1), k_1 =\xi_{\sigma(j_2)}(h_1),  \sigma(\ell_2) =\eta_{h_2}(\sigma(j_2))
\end{equation}
and therefore
\begin{equation}\label{5.45}
R_{32}R_{41}R_{42} R_{43} R_{12} (g_1, g_2, \sigma(g_1), \sigma(g_2))= (k_1,\ell_2, \sigma(k_1), \sigma(\ell_2)).
\end{equation}
Finally, from the assumption that $g_1\in G^{\prime},$ we can show that $k_1\in G^{\prime},$ hence we have the factorization
\begin{equation}\label{5.46}
\sigma(k_1) k_1 =\ell_1\sigma(\ell_1), \quad \ell_1 = \xi_{\sigma(k_1)}(k_1).
\end{equation}
Therefore, when we apply $R_{31}$ to both sides of \eqref{5.45}, we obtain
\begin{equation}\label{5.47}
R_{31}R_{32}R_{41}R_{42} R_{43} R_{12} (g_1, g_2, \sigma(g_1), \sigma(g_2))=(\ell_1, \ell_2, \sigma(\ell_1), \sigma(\ell_2)).
\end{equation} 
In a similar way, we can show that the second line in \eqref{5.39} is well defined under the assumption that $(g_1, g_2) \in G^{(2)}_{\sigma}.$
Successively, we have
\begin{equation}\label{5.48}
\begin{aligned}
& R_{43}R_{12}R_{42}R_{32}R_{41}R_{31} (g_1, g_2, \sigma(g_1),\sigma(g_2))\\
= & R_{43}R_{12}R_{42}R_{32}R_{41}(r_1, g_2,\sigma(r_1), \sigma(g_2))\\
= & R_{43}R_{12}R_{42}(s_1, s_2, \sigma(s_1),\sigma(s_2))\\
= & R_{43}R_{12}(s_1, t_2, \sigma(s_1), \sigma(t_2))\\
= & (u_1, u_2, \sigma(u_1), \sigma(u_2)),
\end{aligned}
\end{equation}
where
\begin{equation}\label{5.49}
\begin{aligned}
& \sigma(g_1) g_1 = r_1\sigma(r_1), r_1 =\xi_{\sigma(g_1)}(g_1)\\
& \sigma(g_2) r_1 = s_1 \sigma(s_2), \sigma(r_1)g_2 = s_2\sigma(s_1), s_1 =\xi_{\sigma(g_2)}(r_1), \sigma(s_2) =\eta_{r_1}(\sigma(g_2))\\
& \sigma(s_2) s_2 = t_2\sigma(t_2), t_2 =\xi_{\sigma(s_2)}(s_2)\\
& s_1 t_2 = u_2 u_1, \sigma(t_2)\sigma(s_1) = \sigma(u_1)\sigma(u_2), u_2 = \xi_{s_1}(t_2), u_1 =\eta_{t_2}(s_1).\\
\end{aligned}
\end{equation}
By using  \eqref{5.44}, we have
\begin{equation}\label{5.50}
\ell_2 = \xi_{\sigma(h_1)}(j_2) = \xi_{\sigma(h_1)}\xi_{\sigma(h_2)}(h_2) =\xi_{\sigma(g_1g_2)g_1}(g_2).
\end{equation}
On the other hand, on using \eqref{5.49}, we obtain
\begin{equation}\label{5.51}
u_2 = \xi_{s_1}(t_2)= \xi_{s_1\sigma(s_2)} (s_2) = \xi_{\sigma(g_2) r_1\sigma(r_1)}(g_2)= \xi_{\sigma(g_1 g_2) g_1}(g_2).
\end{equation}
This shows $\ell_2 = u_2.$   Now, on using \eqref{5.40}, \eqref{5.42},\eqref{5.44}, and \eqref{5.46}, we find
\begin{equation}\label{5.52}
\sigma(g_1 g_2) g_1 g_2 = \ell_2\ell_1 \sigma(\ell_2\ell_1).
\end{equation}
Similarly, on using the relations in \eqref{5.49}, we obtain
\begin{equation}\label{5.53}
\sigma(g_1 g_2) g_1 g_2 = u_2 u_1 \sigma (u_2u_1).
\end{equation}
Therefore, on equating \eqref{5.52} and \eqref{5.53}, we conclude that $x = (\ell_2\ell_1)^{-1} (u_2 u_1)$ satisfies the equation $x\sigma(x) =1.$  Consequently, $x=1$
and as $\ell_2 =u_2,$ we must have $\ell_1= u_1$ and this establishes the validity of the relation in \eqref{5.39}.   To conclude the proof, we will deduce
the relation $\bold{B}_1 R_{21} \bold{B}_2 R_{12}(g_1, g_2) = R_{21}\bold{B}_2 R_{12} \bold{B}_1(g_1,g_2),$ $(g_1, g_2)\in G^{(2)}_{\sigma}$ from \eqref{5.39}, by using
its proof.   Thus we have
\begin{equation}\label{5.54}
\begin{aligned}
 \bold{B}_1 R_{21}\bold{B}_2 R_{12} (g_1, g_2)  = &  \bold{B}_1 R_{21} \bold{B}_2 (h_1, h_2)\\
 = & \bold{B}_1 R_{21} (h_1, \sigma(j_2))\\
 = & \bold{B}_1 (k_1, \sigma(\ell_2))\\
 = & (\sigma(\ell_1), \sigma(\ell_2)).\\
 \end{aligned}
 \end{equation}
 In a similar fashion, we find that
 \begin{equation}\label{5.55}
 R_{21} \bold{B}_2 R_{12} \bold{B}_1 (g_1, g_2) = (\sigma(u_1), \sigma(u_2)).
 \end{equation}
 Hence the assertion follows.
\end{proof}
\begin{remark}\label{R: 5.8}
Note that since $R$ is a Yang-Baxter map, the composite $s\circ R$ is a braiding operator.   Thus according to \eqref{5.36} above, the reduction of this
braiding operator to $G^{\prime}(\sigma)$ is smoothly conjugate to $\bold{B}.$
\end{remark}
We now apply the above results to the case where the Poisson Lie group is $K_{\text{rat}},$ for which the involution $\sigma$ is given by the
formula in \eqref{3.9}.   For this example, recall that the definition of $K_{\text{rat}}\ast K_{\text{rat}}$ is given in \eqref{2.38} (we will connect this with
the object in \eqref{5.1} under assumption (A2)), which is an open submanifold of $K_{\text{rat}}\times K_{\text{rat}}$, equipped with the product Poisson structure.   Hence $K_{\text{rat}}\ast K_{\text{rat}}$ is a Poisson submanifold of $K_{\text{rat}}\times K_{\text{rat}}.$
We have to check that  the assumptions in (A1) to (A4) are satisfied.   First of all,  recall from Theorem 2.5 that we have $K_{\text{rat}}^{g^{-1}} = K_{\text{rat}}^{g}$
for $g\in K_{\text{rat}},$ and that we have a left partial group action $\xi$ and a right partial group action $\eta.$   With the definition of $K_{\text{rat}}^{g}$ in \eqref{2.37}, it is clear
that the geometric object constructed in \eqref{5.1} with $G= K_{\text{rat}}$ is in agreement with what we defined in \eqref{2.38}.   Moreover, the validity of the other conditions
under (A2) are clear.   Regarding the conditions under (A3), first of all, the validity of $K_{\text{rat}}^{\xi_{g}(h)} = K_{\text{rat}}^{h}$ and
$K_{\text{rat}}^{\eta_{h}(g)} = K_{\text{rat}}^g$  is a consequence of the fact that  $(\xi_{g}(h)) = (h)$ and $(\eta_{h}(g)) = (g).$  The other condition is
also clear as we have $g_1 \in K_{\text{rat}}^{g_2}$ iff $\text{supp}\,(g_1)\cap \text{supp}\, (g_2) = \emptyset$ and this condition is
symmetric in $g_1$ and $g_2.$   We now come to (A4).   To check the condition in \eqref{5.6}, suppose $\text{supp}\, (g) =\{a_i, \overline{a}_i\}_{i=1}^{\ell},$
$\text{supp}\, (h) = \{b_j, \overline{b}_j\}_{j=1}^{m}.$   Then $\text{supp}\, (\sigma(g))=\{ -a_i, -\overline{a}_i\}_{i=1}^{\ell}$ and
$\text{supp}\,\{(\sigma(h)) =\{-b_j, -\overline{b}_j\}_{j=1}^{m}.$ From this, it is clear that $\text{supp}\, (g)\cap \text{supp}\, (h) =\emptyset$ iff
$\text{supp}\, (\sigma(g))\cap \text{supp}\, (\sigma(h)) =\emptyset.$    To show that  $K_{\text{rat}}^{\prime} : = (\sigma, id_{K_{\text{rat}}})^{-1} (K_{rat}(\sigma)\cap
(K_{\text{rat}}\times K_{\text{rat}})\neq \emptyset,$ simply take a simple element 
$g_{\alpha, P}.$   As we already observed in Proposition 3.1, for $\alpha\in \mathbb{C}\setminus (\mathbb{R}\cup \sqrt{-1} \mathbb{R}),$  $g_{\alpha, P} \in K_{\text{rat}}^{\prime}.$   Moreover, $\eta_{g_{\alpha, P}}(\sigma(g_{\alpha, P})) = \sigma(\xi_{\sigma(g_{\alpha, P})}(g_{\alpha, P})).$    To show that 
\begin{equation}\label{5.56}
\eta_{\sigma(g)}(\sigma(h)) = \sigma(\xi_{g} (h)), \xi_{\sigma(h)}(\sigma(g)) =\sigma(\eta_{h}(g))\,\, \text{for all}\,\, (g,h) \in K_{\text{rat}}\ast K_{\text{rat}},
\end{equation}
we can use Theorem 2.5 (a), according to which the solution of the refactorization problem
\begin{equation}\label{5.57}
\begin{aligned}
& \sigma(h) \sigma(g)  = \xi_{\sigma(h)}(\sigma(g))\eta_{\sigma(g)}(\sigma(h)),\\
& \text{where}\,\, (\eta_{\sigma(g)}(\sigma(h))= (\sigma(h)), (\xi_{\sigma(h)}(\sigma(g)) = (\sigma(g))\\
\end{aligned}
\end{equation}
is unique and the fact that $(\sigma(\xi_{g}(h))) = (\sigma(h))$ and $(\sigma(\eta_{h}(g))) = (\sigma(g)).$
Let
\begin{equation}\label{5.58}
K_{\text{rat},\sigma}^{(2)} =\{ (g_1, g_2)\in (K_{\text{rat}}^{\prime}\times K_{\text{rat}}^{\prime})\cap(K_{\text{rat}}\ast K_{\text{rat}})\mid (g_1, \sigma(g_2))\in 
K_{\text{rat}}\ast K_{\text{rat}}\}.
\end{equation}
Take $(g_1, g_2)\in K_{\text{rat},\sigma}^{(2)}$ and going through the proof in Theorem 5.7 (c) and making use of the same notations there,
we just have to check that if we take $x= (\ell_2 \ell_1)^{-1} (u_2 u_1)\in K_{\text{rat}},$ then we must have $x=I.$
\begin{prop}\label{P: 5.9}
The element $x= (\ell_2 \ell_1)^{-1} (u_2 u_1)\in K_{\text{rat}}$ which satisfies the equation $x\sigma(x) =I$  is the $n\times n$ identity matrix $I.$
\end{prop}
\begin{proof}
We will keep track of the divisor structure of the various factors which appear in the refactorization problems.   First of all, we have
\begin{equation}\label{5.59}
\begin{aligned}
& (h_1) = (g_1), (h_2) = (g_2),  (j_2) = (h_2),  (\sigma(j_2)) =(\sigma(h_2)),\\
& (k_1) = (h_1), (\sigma(\ell_2)) = (\sigma(j_2)), (\ell_2) =(j_2)\\
& (\ell_1) = (k_1), (\sigma(\ell_1)) =(\sigma(k_1)).\\
\end{aligned}
\end{equation}
From this, we find that
\begin{equation}\label{5.60}
(\ell_1) =(k_1) =(h_1)= (g_1),\,\, (\ell_2) = (j_2)=(h_2)= (g_2).
\end{equation}
Similarly, from
\begin{equation}\label{5.61}
\begin{aligned}
& (r_1) = (g_1), (\sigma(r_1)) = (\sigma(g_1)), (s_1) = (r_1), (\sigma(s_2)) = (\sigma(g_2)),\\
& (t_2) = (s_2),  (\sigma(t_2)) = (\sigma(s_2)), (u_1) = (s_1), \\
& (u_2) =(t_2),
\end{aligned}
\end{equation}
we find
\begin{equation}\label{5.62}
(u_1) = (s_1)= (r_1)= (g_1),\,\, (u_2) = (t_2)=(s_2) =(g_2).
\end{equation} 
Since $\text{supp}(g_1)\cap \text{supp}(g_2) =\emptyset,$ we have
\begin{equation}\label{5.63}
(u_2 u_1)_{0} = (g_1)_{0} + (g_2)_{0},\,\, (u_2u_1)_{\infty} = \overline{(g_1)}_{0}+ \overline{(g_2)}_{0}.
\end{equation}
Likewise,
\begin{equation}\label{5.64}
((\ell_2\ell_1)^{-1})_{0} =  \overline{(g_1)}_{0}+ \overline{(g_2)}_{0},\,\, ((\ell_2\ell_1)^{-1})_{\infty} = (g_1)_{0} + (g_2)_{0}.
\end{equation}
We want to show that $x$ has no poles.   To do so, suppose the contrary, that is, $x$ has poles and zeros.   Let
\begin{equation}\label{5.65}
(x)_{0} =\sum_{i=1}^{d} m_{i}\alpha_{i} +\sum_{i=1}^{d} m_{i}\overline{\alpha}_i + \sum_{j=1}^{e} n_j\beta_j +\sum_{j=1}^{e} n_j\overline{\beta}_j,
\end{equation}
where
\begin{equation}\label{5.66}
\begin{aligned}
& \sum_{i=1}^{d} m_{i}\alpha_{i} +\sum_{i=1}^{d} m_{i}\overline{\alpha}_i \leq (g_1)_{0} + \overline{(g_1)}_{0}, \\
& \sum_{j=1}^{e} n_j\beta_j +\sum_{j=1}^{e} n_j\overline{\beta}_j\leq (g_2)_{0} + \overline{(g_2)}_{0}.\\
\end{aligned}
\end{equation}
Then from the definition of $\sigma,$ we have
\begin{equation}\label{5.67}
(\sigma(x)^{-1})_{0} = \sum_{i=1}^{d} m_{i}(-\alpha_{i}) +\sum_{i=1}^{d} m_{i}(-\overline{\alpha}_i )+ \sum_{j=1}^{e} n_j(-\beta_j) +\sum_{j=1}^{e} n_j(-\overline{\beta}_j).
\end{equation}
Since $g_1, g_2\in K^{\prime}_{\text{rat}},$ we have $\text{supp}(g_i)\cap \text{supp}(\sigma(g_i))=\emptyset,$ $i=1,2.$  Hence the following conditions hold:
\begin{equation}\label{5.68}
\{\alpha_i , \overline{\alpha}_i\}_{i=1}^{d} \cap \{ -\alpha_i, -\overline{\alpha}_i\}_{i=1}^{d}=\emptyset,\,\,
\{\beta_j , \overline{\beta}_j\}_{j=1}^{e} \cap \{ -\beta_j, -\overline{\beta}_j\}_{j=1}^{e}=\emptyset.
\end{equation}
In view of these conditions, it follows from the equation $x = \sigma(x)^{-1}$ and \eqref{5.65}, \eqref{5.67} that
\begin{equation}\label{5.69}
\sum_i m_i \alpha_i + m_{i}\overline{\alpha}_i = \sum_j n_j(-\beta_j) +\sum_j n_j(-\overline{\beta}_j).
\end{equation}
But this is a contradiction to the assumption that $(g_1, \sigma(g_2))\in K_{\text{rat}}\ast K_{\text{rat}}.$
Consequently, $x$ has no poles, and the only such element in $K_{\text{rat}}$ is the identity matrix $I.$
\end{proof}
Now let $\pi_{K_{\text{rat}}}$ be the induced Poisson structure on $K_{\text{rat}}$ as a Poisson Lie subgroup of $(K, \{\cdot,\cdot\}_{J}).$  Recall that
$K_{\text{rat}}\times K_{\text{rat}}$ is equipped with the product Poisson structure, and the open submanifold $K_{\text{rat}}\ast K_{\text{rat}}$ with
the induced structure.  Likewise, we will equip the open submanifold $K^{\prime}_{\text{rat}}$ of $K_{\text{rat}}$ with the induced structure.
We next check that $\sigma$ is a Poisson involution.
\begin{prop}\label{P: 5.10}
The map $\sigma$ defined in \eqref{3.9} is a Poisson involution, when $K_{\text{rat}}$ is equipped with $\pi_{K_{\text{rat}}}.$
\end{prop}
\begin{proof}
Extend $\sigma$ to the full group $K$ using the same formula, denote the extension by $\sigma_e,$ and let $\iota_{K_{\text{rat}}}: K_{\text{rat}}\longrightarrow K$
be the inclusion map.   In view of the relation $\iota_{K_{\text{rat}}}\circ \sigma = \sigma_e\circ \iota_{K_{\text{rat}}}$ and the fact that 
$\iota_{K_{\text{rat}}}$ is a Poisson map, it suffices to show that $\sigma_e: (K, \{\cdot, \cdot\}_{J}) \longrightarrow (K, \{\cdot, \cdot\}_{J})$ is a Poisson map.
So let $\varphi, \psi\in \mathcal{F}(K).$  Then by a direct calculation, for $g\in K,$ we have
\begin{equation}\label{5.70}
\begin{aligned}
& D(\varphi\circ \sigma_e)(g)(z) = U^* D^{\prime}\varphi(\sigma_e(g))(-z)U\\
& D^{\prime}(\varphi\circ \sigma_e)(g)(z) = U^* D\varphi(\sigma_e(g))(-z)U.\\
\end{aligned}
\end{equation}
Therefore, by the definition of $J$ in \eqref{2.45}, we have
\begin{equation}\label{5.71}
\begin{aligned}
J(D(\varphi\circ \sigma_e)(g))(z) = U^* J(D^{\prime}\varphi(\sigma_e(g))(-z) U,\\
J(D^{\prime}(\varphi\circ \sigma_e)(g))(z) = U^* J(D\varphi(\sigma_e(g))(-z) U.\\
\end{aligned}
\end{equation}
Hence on using the pairing $(\cdot, \cdot)_{\fk}$ in \eqref{2.44}, we find that
\begin{equation}\label{5.72}
\begin{aligned}
& (J(D(\varphi\circ \sigma_e)(g)), D(\psi\circ \sigma_e)(g))_{\fk} =  -(J(D^{\prime}\varphi(\sigma_e(g)), D^{\prime}\psi(\sigma_e(g))))_{\fk},\\
&  (J(D^{\prime}(\varphi\circ \sigma_e)(g)), D^{\prime}(\psi\circ \sigma_e)(g))_{\fk} = -  (J(D\varphi(\sigma_e(g)), D\psi(\sigma_e(g))))_{\fk}\\                                                                                  
\end{aligned}
\end{equation}
and the assertion that $\sigma_e$ is Poisson follows from this formula.   
\end{proof}
From this proposition, we can now conclude that the map $\Sigma$ in \eqref{5.21} with $G= K_{\text{rat}}$   is a Poisson involution by Lemma 5.5 and that its stable locus is
given by $K_{\text{rat}}(\sigma),$ the graph of $\sigma.$   This is a Dirac submanifold of $K_{\text{rat}}\times K_{\text{rat}}$ and its induced
Poisson structure $\pi_{K_{\text{rat}(\sigma)}}$ is related to the structure $\pi_{K_{\text{rat}}}$ through the relation in \eqref{5.22}.  Consequently,
the open submanifold $K^{\prime}_{\text{rat}}(\sigma)$ carries an induced structure $\pi_{K^{\prime}_{\text{rat}}(\sigma)}.$
We are now ready to state the following consequence of Theorem 5.7.

\begin{coro}\label{C:5.11}
(a)  The map $R$ commutes with $\Sigma\mid_{K_{\text{rat}}\ast K_{\text{rat}},}$ i.e. $\Sigma\mid_{K_{\text{rat}}\ast K_{\text{rat}}}\circ R = 
R\circ \Sigma\mid_{K_{\text{rat}}\ast K_{\text{rat}}}$ and therefore its
reduction
\begin{equation}\label{5.73}
R_{\text{red}} : K_{\text{rat}}^{\prime}(\sigma)\longrightarrow K_{\text{rat}}^{\prime}(\sigma)
\end{equation}
is a Poisson diffeomorphism, when $K_{\text{rat}}^{\prime}(\sigma)$ is equipped with any nonzero multiple of $\pi_{K_{\text{rat}}^{\prime}(\sigma)}.$
\newline
(b)  The map
\begin{equation}\label{5.74} 
\bold{B} : (K_{\text{rat}}^{\prime}, \pi_{K_{\text{rat}}^{\prime}}) \longrightarrow (K_{\text{rat}}^{\prime}, \pi_{K_{\text{rat}}^{\prime}})
\end{equation}
satisfies the relation
\begin{equation}\label{5.75}
\bold{B} = ((\sigma, id_{K_{\text{rat}}})^{\prime})^{-1}\circ (s\circ R)_{\text{red}} \circ (\sigma, id_{K_{\text{rat}}})^{\prime}
\end{equation}
and hence is a Poisson diffeomorphism.  Here $(s\circ R)_{\text{red}} : K_{\text{rat}}^{\prime}(\sigma)\longrightarrow K_{\text{rat}}^{\prime}(\sigma)$ is
the reduction of $s\circ R$ to $K_{\text{rat}}^{\prime}(\sigma).$
\newline
(c)  The map $\bold{B}$ satisfies the reflection equation
\begin{equation}\label{5.76}
\bold{B}_1 R_{21} \bold{B}_2 R_{12} = R_{21} \bold{B}_2 R_{12} \bold{B}_1
\end{equation}
where we interpret \eqref{5.76} as an equality of maps from $K_{\text{rat},\sigma}^{(2)}$ to itself, where
$K_{\text{rat},\sigma}^{(2)}$ is defined in \eqref{5.58}.  Hence $\bold{B}$ is a reflection map.

\end{coro}
%%%%%%%%%%%%%%%%%%%%%%%%%%%%%%%%%%%%%%%%%%%%%%%%%%%%%%%%%%%%%%%%%%%%%%%%%%%%%%%%%%%%%
\section{Conclusion}
%%%%%%%%%%%%%%%%%%%%%%%%%%%%%%%%%%%%%%%%%%%%%%%%%%%%%%%%%%%%%%%%%%%%%%%%%%%%%%%%%%%%

This work was motivated by the soliton-boundary interaction process for the $n$-Manakov system on the half-line,
with mixed Dirichlet/Neumann boundary conditions at $x=0$, as described in \cite{CZ2}.    We have taken a first step
here in (a) constructing reflection maps from Yang-Baxter maps on various geometrical objects and understanding
their relationships,  (b) describing the symplectic and Poisson geometry of such maps.   Thus we have proved here, for
the first time, the symplectic/Poisson properties of reflection maps.   As is clear from our work in the previous sections,
an involution plays an important role on each level, this is a structure which emerges in the authors' use of the nonlinear
mirror image method in \cite{CZ2}.    In this concluding section, we will give a short discussion of several issues which
we have not addressed in this work, as well as making some comments on the significance of our findings.

As in the case of the polarization scattering map in \cite{L1}, the $N$-body polarization reflection map in Section 4 should
be regarded as a component of the full scattering map, which would include as its components the map
which gives the change in asymptotic velocities and the map which gives the change in phase
shifts.  We hope to extend our results here to that of the full scattering map, as well as the construction of action-angle
variables on multi-soliton manifolds for the half-line problem.   This latter endeavour,  of course, would involve the
presence of the soliton parameters $\alpha_i$  in the symplectic form, as they are part of the scattering data in the reflectionless case. 

On the other hand, although we are focusing our 
attention here to the $n$-Manakov system, however, it is clear that the same methodology can be adapted to other multi-component integrable soliton equations on the half-line, if the nonlinear mirror image method applies, an important ingredient being the existence of an involution which can be extended to a Lie group anti-morphism.   This is in fact one of the motivations behind formulating several results in an abstract way in Section 5, when we deal with reflection maps on Poisson Lie groups, as different multi-component integrable soliton equations correspond to different Lie groups.   In this connection,
let us also recall that in \cite{L1}, the author shows that if we denote the Poisson Lie group dual to $(K, \{\cdot, \cdot\}_{J})$
by $K_{J},$ and the dressing orbit of $K_{J}$ through $g_{\alpha, E_k}$ by $\mathcal{L}(\alpha, E_k)$,
then the map $R\mid\mathcal{L}(\alpha_1, E_k)\times \mathcal{L}(\alpha_2, E_{\ell})$  (where $R$ is given in \eqref{2.41})
 and the map $R^{k,\ell}(\alpha_1, \alpha_2)$ in \eqref{2.18} are conjugate to each other (see \eqref{5.68} in  \cite{L1}).
 Thus from this point of view, we could have developed our results in Section 3 and Section 4 of our present work starting
 with the results in Section 5.  But of course this would be unnecessarily complicated.   The point we are trying to make here is that the Poisson Lie group carries the complete information, as there are various dressing orbits of $K_{J}$ which
 could be of interest in the study of higher order multi-soliton solutions (the ones in \cite{CZ2} correspond to Riemann-Hilbert problems with distinct simple zeros).    
 
 In any case, extending the results in Section 3 and Section 4 to the case of Poisson
 Lie groups via the method of Dirac reduction is of intrinsic geometric interest.   Here we recall the work in \cite{LYZ}, in which they show how to construct a solution of the Yang-Baxter equation on a group, assuming the existence of a pair of actions satisfying a compatibility condition.   This result is purely algebraic and in particular is devoid of any meaning in Poisson geometry.    In Theorem 5.4, by following the same argument which was used in the proof of Theorem 5.13 in \cite{L1} for  the case of $K_{\text{rat}},$ we show how two compatible partial actions on a Poisson Lie group $G$ can give rise to a Yang-Baxter map $R$ which is also a Poisson diffeomorphism.   And then by postulating the existence of  a Poisson involution $\sigma$
 on $G$ which is also a Lie group anti-morphism satisfying some additional conditions, we can define a reflection map
 $\bold{B}$ which is also a Poisson diffeomorphism.    And the method we use provides
 another illustration of the use of Dirac reduction, which was first developed in \cite{L2} in order to understand a class of spin Calogero-Moser systems associated with symmetric Lie subalgebras, and the spin-generalized Ruisjenaars-Schneider equations which correspond to $N$-soliton solutions of $A_n^{(1)}$ affine Toda field theory \cite{BHO}.   
 
 We hope to have a better understanding of the integrability of the various reflection maps in this work in the future.

\vskip .1in
\noindent  \textbf{Acknowledgments}  Luen-Chau Li acknowledges the support from the Simons Foundation through grants \#278994,  \#585813,  the London 
Mathematical Society research in pairs grant \#41849, and the Research Visitor Centre of the School of Mathematics, University of Leeds.    
Vincent Caudrelier acknowledges the support from the Simons Foundation, the London Mathematical Society research in pairs grant \#41849,
and the Shapiro visitor program in mathematics of the 
Pennsylvania State University.  Both authors would like to thank the host institutions for the hospitality during their visits.
Finally, the authors are thankful to the referee for his/her careful reading of the manuscript and making constructive
comments for making this work more readable.
\vskip .1in
\noindent  \textbf{Declarations} 
\vskip .1in
\noindent \textbf{Conflict of interest}  The authors declare that there is no conflict of interest.

%%%%%%%%%%%%%%%%%%%%%%%%%%%%%%%%%%%%%%%%%%%%%%%%%%%%%%%%%%%%%%%%%%%%%%%%%%%%%%%%%%%%%%%%%

\end{document}